\documentclass[a4paper,UKenglish]{lipics-v2016}

\usepackage[utf8]{inputenc}
\usepackage{mathscinet}
\usepackage{comment}
\usepackage{graphicx}
\usepackage{booktabs}
\usepackage{hyperref}
\usepackage{pgfkeys}
\usepackage[numbers,sort&compress]{natbib}
\usepackage{enumerate}
\usepackage{verbatim}
\usepackage{nicefrac}
\usepackage{framed}
\usepackage{paralist}
\usepackage[inline]{enumitem}

\usepackage[dvipsnames]{xcolor}
\usepackage{tikz}
\usetikzlibrary{fit}
\usetikzlibrary{intersections}
\usetikzlibrary{decorations.pathreplacing}
\usetikzlibrary{calc}

\usepackage{amsmath}
\usepackage{amssymb}
\usepackage{thmtools}
\usepackage{thm-restate}

\usepackage{xspace}
\usepackage{authblk}
\usepackage[color=green!20]{todonotes}

\def\wcli{\textsc{Maximum Weighted Clique}\xspace}
\def\cli{\textsc{Maximum Clique}\xspace}
\def\mis{\textsc{Maximum Independent Set}\xspace}
\def\mwis{\textsc{Maximum Weighted Independent Set}\xspace}

\tikzstyle{vp}=[circle,fill,inner sep=0pt, minimum size=0.1cm]
\tikzstyle{vps}=[circle,fill,inner sep=0pt, minimum size=0.065cm]

\def\centerarc[#1](#2)(#3:#4:#5)
    { \draw[#1] ($(#2)+({#5*cos(#3)},{#5*sin(#3)})$) arc (#3:#4:#5); }

\def\seg{\text{seg}}

\DeclareMathOperator{\ocp}{ocp}

\title{QPTAS and Subexponential Algorithm for Maximum Clique on Disk Graphs\footnote{Research partially supported by EPSRC grant FptGeom (EP/N029143/1) and ANR grant ESIGMA (ANR-17-CE40-0028)}}

\author[1]{\'Edouard Bonnet}
\author[1]{Panos Giannopoulos}
\author[2]{Eun Jung Kim}
\author[3]{Pawe\l{} Rz\k{a}\.zewski}
\author[2]{Florian Sikora}
\affil[1]{Department of Computer Science, Middlesex University, London
  \texttt{edouard.bonnet@dauphine.fr}, \texttt{p.giannopoulos@mdx.ac.uk}}
\affil[2]{Universit\'{e} Paris-Dauphine, PSL Research University, CNRS UMR, LAMSADE, Paris, France\\
\texttt{\{eun-jung.kim,florian.sikora\}@dauphine.fr}}
\affil[3]{Faculty of Mathematics and Information Science,\\ 
  Warsaw University of Technology \texttt{p.rzazewski@mini.pw.edu.pl}}

\authorrunning{\'E. Bonnet and P. Giannopoulos and E.~J.~Kim and P. Rz\k{a}\.zewski and F. Sikora}
\Copyright{\'Edouard Bonnet, Panos Giannopoulos, Eun Jung Kim, Pawe\l{} Rz\k{a}\.zewski, and Florian Sikora}
\subjclass{G.2.2 Graph Theory, F.2.2 Nonnumerical Algorithms and Problems}
\keywords{disk graph, maximum clique, computational complexity}

\begin{document}
 
\maketitle

\begin{abstract}
  A (unit) disk graph is the intersection graph of closed (unit) disks in the plane.
  Almost three decades ago, an elegant polynomial-time algorithm was found for \textsc{Maximum Clique} on unit disk graphs [Clark, Colbourn, Johnson; Discrete Mathematics '90].
  Since then, it has been an intriguing open question whether or not tractability can be extended to general disk graphs.
  We show the rather surprising structural result that a disjoint union of cycles is the complement of a disk graph if and only if at most one of those cycles is of odd length.
  From that, we derive the first QPTAS and subexponential algorithm running in time $2^{\tilde{O}(n^{2/3})}$ for \textsc{Maximum Clique} on disk graphs.
  In stark contrast, \textsc{Maximum Clique} on intersection graphs of filled ellipses or filled triangles is unlikely to have such algorithms, even when the ellipses are close to unit disks.
  Indeed, we show that there is a constant ratio of approximation which cannot be attained even in time $2^{n^{1-\varepsilon}}$, unless the Exponential Time Hypothesis fails.
\end{abstract}

\section{Introduction}
An \emph{intersection graph} of geometric objects has one vertex per object and an edge between every pair of vertices corresponding to intersecting objects. 
Intersection graphs for many different families of geometric objects have been studied due to their practical applications and rich structural properties~\cite{McKee1999, Brandstadt1999}.
Among the most studied ones are \emph{disk graphs}, which are intersection graphs of closed disks in the plane, and their special case, \emph{unit disk graphs}, where all the radii are the same. Their applications range from sensor networks to map labeling~\cite{DBLP:conf/waoa/Fishkin03}, and many standard optimization problems have been studied on disk graphs, see for example~\cite{EJvL2009} and references therein.
In this paper, we study \textsc{Maximum Clique} on general disk graphs. 

\paragraph*{Known results.}
Recognizing unit disk graphs is NP-hard \cite{Breu98}, and even $\exists \mathbb{R}$-complete~\cite{Kang12}.
Clark et al.~\cite{Clark90} gave a polynomial-time algorithm for \cli on unit disk graphs with a geometric representation. 
The core idea of their algorithm can actually be adapted so that the geometric representation is no longer needed \cite{Raghavan03}. 
The complexity of the problem on general disk graphs is unfortunately still unknown. 
Using the fact that the transversal number for disks is $4$, Amb\"uhl and Wagner~\cite{Ambuhl05} gave a simple $2$-approximation algorithm for \cli on general disk graphs. 
They also showed the problem to be APX-hard on intersection graphs of ellipses and gave a $9\rho^2$-approximation algorithm for filled ellipses of aspect ratio at most $\rho$. 
Since then, the problem has proved to be elusive with no new positive or negative results.
The question on the complexity and further approximability of \cli on general disk graphs is considered as folklore~\cite{bang2006}, but was also explicitly mentioned as an open problem by Fishkin~\cite{DBLP:conf/waoa/Fishkin03}, Amb\"uhl and Wagner~\cite{Ambuhl05} and Cabello~\cite{CabelloOpen,Cabello2015}.

A closely related problem is \textsc{Maximum Independent Set}, which is known to be W[1]-hard (even on unit disk graphs \cite{Marx08}) and to admit a subexponential exact algorithm~\cite{AlberF04} and PTAS~\cite{Erlebach2005,Chan2003} on disk graphs.

\paragraph*{Results and organization.}
In Section~\ref{sec:structural}, we mainly prove that the disjoint union of two odd cycles is not the complement of a disk graph.
To the best of our knowledge, this is the first structural property that general disk graphs do not inherit from strings or from convex objects.
We provide an infinite family of forbidden induced subgraphs, an analogue to the recent work of Atminas and Zamaraev on unit disk graphs~\cite{Atminas16}.
In Section~\ref{sec:algorithms}, we show how to use this structural result to approximate and solve \mis on complements of disk graphs, hence \cli on disk graphs.
More precisely, we present the first quasi-polynomial-time approximation scheme (QPTAS) and subexponential-time algorithm for \cli on disk graphs, even without the geometric representation of the graph.
In Section~\ref{sec:gen&lim}, we highlight how those algorithms contrast with the situation for ellipses or triangles, where there is a constant $\alpha>1$ for which an $\alpha$-approximation running in subexponential time is highly unlikely (in particular, ruling out at once QPTAS \emph{and} subexponential-time algorithm).
We conclude in Section~\ref{sec:perspectives} with a few open questions.

\paragraph*{Definitions and notations.}
For two integers $i \leqslant j$, we denote by $[i,j]$ the set of integers $\{i,i+1,\ldots, j-1, j\}$.
For a positive integer $i$, we denote by $[i]$ the set of integers $[1,i]$.
If $S$ is a subset of vertices of a graph, we denote by $N(S)$ the open neighborhood of $S$ and by $N[S]$ the set $N(S) \cup S$.
The \emph{2-subdivision} of a graph $G$ is the graph $H$ obtained by subdividing each edge of $G$ exactly twice.
If $G$ has $n$ vertices and $m$ edges, then $H$ has $n+2m$ vertices and $3m$ edges.
The \emph{co-2-subdivision} of $G$ is the complement of $H$.
Hence it has $n+2m$ vertices and ${n+2m \choose 2} - 3m$ edges.
The \emph{co-degree} of a graph is the maximum degree of its complement.
A \emph{co-disk} is a graph that is the complement of a disk graph.

For two distinct points $x$ and $y$ in the plane, we denote by $\ell(x,y)$ the unique line going through $x$ and $y$, and by $\seg(x,y)$ the closed straight-line segment whose endpoints are $x$ and $y$.
If $s$ is a segment with positive length, then we denote by $\ell(s)$ the unique line containing $s$.
We denote by $d(x,y)$ the euclidean distance between points $x$ and $y$.
We will often define disks and elliptical disks by their boundary, i.e., circles and ellipses, and also use the following basic facts.
There are exactly two circles that pass through a given point with a given tangent at this point and a given radius; one if we further specify on which side of the tangent the circle is.
There is exactly one circle which passes through two points with a given tangent at one of the two points, provided the other point is \emph{not} on this tangent.
Finally, there exists one (not necessarily unique) ellipse which passes through two given points with two given tangents at those points. 

The \emph{Exponential Time Hypothesis} (ETH) is a conjecture by Impagliazzo et al. asserting that there is no $2^{o(n)}$-time algorithm for \textsc{3-SAT} on instances with $n$ variables \cite{ImpagliazzoETH}. 
The ETH, together with the sparsification lemma \cite{ImpagliazzoETH}, even implies that there is no $2^{o(n+m)}$-time algorithm solving \textsc{3-SAT}.

\section{Disk graphs with co-degree 2}\label{sec:structural}

In this section, we fully characterize the degree-2 complements of disk graphs.
We show the following:

\begin{theorem}\label{thm:main-structural}
  A disjoint union of paths and cycles is the complement of a disk graph if and only if the number of odd cycles is at most one. 
\end{theorem}

We split this theorem into two parts.
In the first one, Section~\ref{subsec:notco-disk}, we show that the union of two disjoint odd cycles is not the complement of a disk graph.
This is the part that will be algorithmically useful.
As disk graphs are closed under taking induced subgraphs, it implies that in the complement of a disk graph two vertex-disjoint odd cycles have to be linked by at least one edge.
This will turn out useful when solving \mis on the complement of the graph (to solve \cli on the original graph).
In the second part, Section~\ref{subsec:co-disk}, we show how to represent the complement of the disjoint union of even cycles and exactly one odd cycle.
Although this result is not needed for the forthcoming algorithmic section, it nicely highlights the singular role that parity plays and exposes the complete set of disk graphs of co-degree 2. 

\subsection{The disjoint union of two odd cycles is not co-disk}
\label{subsec:notco-disk}

We call \emph{positive distance} between two non-intersecting disks the minimum of $d(x,y)$ where $x$ is in one disk and $y$ is in the other.
If the disks are centered at $c_1$ and $c_2$ with radius $r_1$ and $r_2$, respectively, then this value is $d(c_1,c_2)-r_1-r_2$.
We call \emph{negative distance} between two intersecting disks the length of the straight-line segment defined as the intersection of three objects: the two disks and the line joining their center.
This value is $r_1+r_2-d(c_1,c_2)$, which is positive.

We call \emph{proper representation} a disk representation where every edge is witnessed by a proper intersection of the two corresponding disks, i.e., the interiors of the two disks intersect.
It is easy to transform a disk representation into a proper representation (of the same graph). 

\begin{lemma}\label{lem:proper}
  If a graph has a disk representation, then it has a proper representation.
\end{lemma}

\begin{proof}
  If two disks intersect non-properly, we increase the radius of one of them by $\varepsilon/2$ where $\varepsilon$ is the smallest positive distance between two disks.
\end{proof}

In order not to have to discuss about the corner case of three aligned centers in a disk representation, we show that such a configuration is never needed to represent a disk graph. 

\begin{lemma}\label{lem:generalPosition}
  If a graph has a disk representation, it has a proper representation where no three centers are aligned.
\end{lemma}

\begin{proof}
  By Lemma~\ref{lem:proper}, we have or obtain a proper representation.
  Let $\varepsilon$ be the minimum between the smallest positive distance and the smallest negative distance.
  As the representation is proper, $\varepsilon > 0$.
  If three centers are aligned, we move one of them to any point which is not lying in a line defined by two centers in a ball of radius $\varepsilon/2$ centered at it.
  This decreases by at least one the number of triple of aligned centers, and can be repeated until no three centers are aligned.
\end{proof}

From now on, we assume that every disk representation is proper and without three aligned centers.
We show the folklore result that in a representation of a $K_{2,2}$ that sets the four centers in convex position, both non-edges have to be \emph{diagonal}.

\begin{lemma}\label{lem:cotwoKtwo}
  In a disk representation of $K_{2,2}$ with the four centers in convex position, the non-edges are between vertices corresponding to opposite centers in the quadrangle. 
\end{lemma}
\begin{proof}
  Let $c_1$ and $c_2$ be the centers of one non-edge, and $c_3$ and $c_4$ the centers of the other non-edge.
  Let $r_i$ be the radius associated to center $c_i$ for $i \in [4]$.
  It should be that $d(c_1,c_2)>r_1+r_2$ and $d(c_3,c_4)>r_3+r_4$ (see Figure~\ref{fig:k22}).
  Assume $c_1$ and $c_2$ are consecutive on the convex hull formed by $\{c_1, c_2, c_3, c_4\}$, and say, without loss of generality, that the order is $c_1, c_2, c_3, c_4$.
  Let $c$ be the intersection of $\seg(c_1,c_3)$ and $\seg(c_2,c_4)$.
  It holds that $d(c_1,c_3) + d(c_2,c_4) = d(c_1,c) + d(c,c_3) + d(c_2,c) + d(c,c_4) = (d(c_1,c)+d(c,c_2)) + (d(c_3,c)+d(c,c_4)) > d(c_1,c_2) + d(c_3,c_4) > r_1 + r_2 + r_3 + r_4 = (r_1+r_3)+(r_2+r_4)$.
  Which implies that $d(c_1,c_3) > r_1+r_3$ or $d(c_2,c_4) > r_2+r_4$; a contradiction.
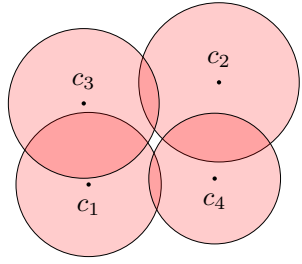
\begin{figure}
\centering
\begin{tikzpicture}[scale=0.25,
    dot/.style={fill,circle,inner sep=-0.02cm}
  ]
\draw[very thin,fill=red, fill opacity=0.2] (11.8405399558, 12.6637388353) circle (3.82299906735);
\node[dot] at (11.8405399558, 12.6637388353) {};
\node at (11.8405399558, 12.6637388353 - 1.2) {$c_1$};
\draw[very thin,fill=red, fill opacity=0.2] (18.7026380772, 18.0821427855) circle (4.22302853751);
\node[dot] at (18.7026380772, 18.0821427855) {};
\node at (18.7026380772, 18.0821427855 + 1.2) {$c_2$};
\draw[very thin,fill=red, fill opacity=0.2] (11.5839455246, 16.9641945052) circle (3.96909254758);
\node[dot] at (11.5839455246, 16.9641945052) {};
\node at (11.5839455246, 16.9641945052 + 1.2) {$c_3$};
\draw[very thin,fill=red, fill opacity=0.2] (18.4743856348, 12.97903103) circle (3.46657454294);
\node[dot] at (18.4743856348, 12.97903103) {};
\node at (18.4743856348, 12.97903103 - 1.2) {$c_4$};
\end{tikzpicture}
\caption{Disk realization of a $K_{2,2}$. As the centers are positioned, it is impossible that the two non-edges are between the disks 2 and 3, and between the disks 1 and 4 (or between the disks 1 and 3, and between the disks 2 and 4).}\label{fig:k22}
\end{figure}
\end{proof}

We derive a useful consequence of the previous lemma, phrased in terms of intersections of lines and segments.

\begin{corollary}\label{cor:intersect}
In any disk representation of $K_{2,2}$ with centers $c_1, c_2, c_3, c_4$ with the two non-edges between the vertices corresponding to $c_1$ and $c_2$, and between $c_3$ and $c_4$, it should be that $\ell(c_1,c_2)$ intersects $\seg(c_3,c_4)$ or $\ell(c_3,c_4)$ intersects $\seg(c_1,c_2)$.    
\end{corollary}
\begin{proof}
  Either the disk representation has the four centers in convex position.
  Then, by Lemma~\ref{lem:cotwoKtwo}, $\seg(c_1,c_2)$ and $\seg(c_3,c_4)$ are the diagonals of a convex quadrangle.
  Hence they intersect, and \emph{a fortiori}, $\ell(c_1,c_2)$ intersects $\seg(c_3,c_4)$ ($\ell(c_3,c_4)$ intersects $\seg(c_1,c_2)$, too).

  Or the disk representation has one center, say without loss of generality, $c_1$, in the interior of the triangle formed by the other  three centers.
  In this case, $\ell(c_1,c_2)$ intersects $\seg(c_3,c_4)$.
  If instead a center in $\{c_3,c_4\}$ is in the interior of the triangle formed by the other centers, then $\ell(c_3,c_4)$ intersects $\seg(c_1,c_2)$.
\end{proof}

We can now prove the main result of this section thanks to the previous corollary, parity arguments, and \emph{some elementary properties of closed plane curves}, namely Property I and Property III of the eponymous paper \cite{Tait1877}.

\begin{theorem}\label{thm:main-structural-non-disk} 
  The complement of the disjoint union of two odd cycles is not a disk graph.
\end{theorem}
\begin{proof}
  Let $s$ and $t$ be two positive integers and $G=\overline{C_{2s+1} + C_{2t+1}}$ the complement of the disjoint union of a cycle of length $2s+1$ and a cycle of length $2t+1$.
  Assume that $G$ is a disk graph.
  Let $\mathcal C_1$ (resp. $\mathcal C_2$) be the cycle embedded in the plane formed by $2s+1$ (resp. $2t+1$) straight-line segments joining the consecutive centers of disks along the first (resp. second) cycle.
 Observe that the segments of those two cycles correspond to the non-edges of $G$.
 We number the segments of $\mathcal C_1$ from $S_1$ to $S_{2s+1}$, and the segments of $\mathcal C_2$, from $S'_1$ to $S'_{2t+1}$.
  
  For the $i$-th segment $S_i$ of $\mathcal C_1$, let $a_i$ be the number of segments of $\mathcal C_2$ intersected by the line $\ell(S_i)$ prolonging $S_i$, let $b_i$ be the number of segments $S'_j$ of $\mathcal C_2$ such that the prolonging line $\ell(S'_j)$ intersects $S_i$, and let $c_i$ be the number of segments of $\mathcal C_2$ intersecting $S_i$.
  For the second cycle, we define similarly $a'_j$, $b'_j$, $c'_j$.
  The quantity $a_i+b_i-c_i$ counts the number of segments of $\mathcal C_2$ which can possibly represent a $K_{2,2}$ with $S_i$ according to Corollary~\ref{cor:intersect}.
  As we assumed that $G$ is a disk graph, $a_i+b_i-c_i = 2t+1$ for every $i \in [2s+1]$.
  Otherwise there would be at least one segment $S'_j$ of $\mathcal C_2$ such that $\ell(S_i)$ does not intersect $S'_j$ \emph{and} $\ell(S'_j)$ does not intersect $S_i$.
  
  Observe that $a_i$ is an even integer since $\mathcal C_2$ is a closed curve.
  Also, $\Sigma_{i=1}^{2s+1}a_i+b_i-c_i=(2t+1)(2s+1)$ is an odd number, as the product of two odd numbers.
  This implies that $\Sigma_{i=1}^{2s+1}b_i-c_i$ shall be odd.
  $\Sigma_{i=1}^{2s+1}c_i$ counts the number of intersections of the two closed curves $\mathcal C_1$ and $\mathcal C_2$, and is therefore even.
  Hence, $\Sigma_{i=1}^{2s+1}b_i$ shall be odd.
  Observe that $\Sigma_{i=1}^{2s+1}b_i=\Sigma_{j=1}^{2t+1}a'_j$ by reordering and reinterpreting the sum from the point of view of the segments of $\mathcal C_2$.
  Since the $a'_j$ are all even, $\Sigma_{i=1}^{2s+1}b_i$ is also even; a contradiction.
\end{proof}

\subsection{The disjoint union of cycles with at most one odd is co-disk}
\label{subsec:co-disk}

We only show the following part of Theorem~\ref{thm:main-structural} to emphasize that, rather unexpectedly, parity plays a crucial role in disk graphs of co-degree 2.
It is also amusing that the complement of any odd cycle is a \emph{unit} disk graph while the complement of any even cycle of length at least 8 is not \cite{Atminas16}.
Here, the situation is somewhat reversed: complements of even cycles are \emph{easier} to represent than complements of odd cycles.

\begin{theorem}\label{thm:coEvenCycles}
The complement of the disjoint union of even cycles and one odd cycle is a disk graph.
\end{theorem}

\begin{proof}
We start with a disk representation of the complement of one even cycle $C_{2s}$.
  Again, this construction is not possible with \emph{unit} disks for even cycles of length at least 8.
  We assume that the vertices of the cycle $C_{2s}$ are $1, 2, \ldots, 2s$ in this order.
  For each $i \in [2s]$, the disk $\mathcal D_i$ encodes the vertex $i$. 
  We start by fixing the disks $\mathcal D_1$, $\mathcal D_2$, and $\mathcal D_{2s}$.
  Those three disks have the same radius.
  We place $\mathcal D_2$ and $\mathcal D_{2s}$ side by side: their centers have the same $y$-coordinate.
  They intersect and the distance between their center is $\varepsilon > 0$.
  We define $\mathcal D_1$ as the disk above $\mathcal D_2$ and $\mathcal D_{2s}$ tangent to those two disks and sharing the same radius.
  We denote by $p_1$ its intersection with $\mathcal D_2$ and by $p_s$ its intersection with $\mathcal D_{2s}$.
  We then slightly shift $\mathcal D_1$ upward so that it does not touch (nor does it intersect) $\mathcal D_2$ and $\mathcal D_{2s}$ anymore.
  While we do this translation, we imagine that the points $p_1$ and $p_s$ remain fixed at the boundary of $\mathcal D_2$ and $\mathcal D_{2s}$ respectively (see Figure~\ref{fig:complement-one-even-cycle1}).
  Let $p_2, p_3, \ldots, p_{s-1}$ points in the interior of $\mathcal D_1$ and below the line $\ell(p_1,p_s)$ such that $p_1, p_2, \ldots, p_{s-1}, p_s$ form an $x$-monotone convex chain (see Figure~\ref{fig:complement-one-even-cycle2}). 
  
  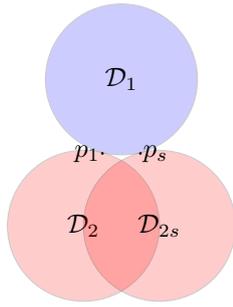
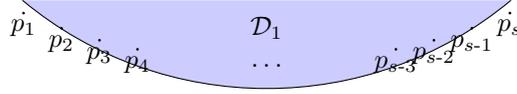
\begin{figure}[h!]
    \centering
    \begin{minipage}{0.3\textwidth}
      \centering
    \begin{tikzpicture}[
        dot/.style={fill,circle,inner sep=-0.01cm},
        vert/.style={draw, fill=red, opacity=0.2},
        verta/.style={draw, fill=blue, opacity=0.2},
      ]
      \def\s{1}
      \coordinate (c2) at (0,0) ;
      \draw[vert] (c2) circle (1) ;
      \node at (c2) {$\mathcal D_2$} ;

      \coordinate (c2s) at (\s,0) ;
      \draw[vert] (c2s) circle (1) ;
      \node at (c2s) {$\mathcal D_{2s}$} ;

      \coordinate (c1) at (\s / 2,2 - \s / 17) ;
      \draw[verta] (c1) circle (1) ;
      \node at (c1) {$\mathcal D_1$} ;

      \node[dot] at (0.25,0.96) {} ;
      \node[dot] at (\s - 0.25,0.96) {} ;

      \node at (0.05,0.96) {$p_1$} ;
      \node at (\s - 0.05,0.96) {$p_s$} ;   
    \end{tikzpicture}
    \subcaption{Three important disks with the same size $\mathcal D_1$, $\mathcal D_2$, $\mathcal D_{2s}$.}
    \label{fig:complement-one-even-cycle1}
    \end{minipage}
  \qquad
  \begin{minipage}{0.6\textwidth}
      \centering
      \begin{tikzpicture}
        [
          dot/.style={fill,circle,inner sep=-0.01cm}
        ]
        \centerarc[draw, fill=blue, fill opacity=0.2](0,0)(-50:-130:5) ;

        \def\xp{-3.2}
        \def\yp{-4}
        \def\xpb{3.2}
        \node[dot] at (\xp,\yp) {} ;
        \node at (\xp,\yp - 0.2) {$p_1$} ;
        \node[dot] at (\xpb,\yp) {} ;
        \node at (\xpb,\yp - 0.2) {$p_s$} ;

        \foreach \k/\i/\j in {2/0.5/0.2,3/1/0.36,4/1.5/0.48}{
          \node[dot] at (\xp+\i,\yp-\j) {} ;
          \node at (\xp+\i,\yp-\j- 0.2) {$p_\k$} ;
        }
        \foreach \k/\i/\j in {1/0.5/0.2,2/1/0.36,3/1.5/0.48}{
          \node[dot] at (\xpb-\i,\yp-\j) {} ;
          \node at (\xpb-\i,\yp-\j- 0.2) {$p_{s\text{-}\k}$} ;
        }
        \node at (0,-4.7) {$\ldots$} ;
        \node at (0,-4.2) {$\mathcal D_1$} ;
      \end{tikzpicture}
  \subcaption{Zoom where $\mathcal D_1$ almost touches $\mathcal D_2$ and $\mathcal D_{2s}$.}
  \label{fig:complement-one-even-cycle2}
  \end{minipage}
   \caption{The disks $\mathcal D_1$, $\mathcal D_2$, $\mathcal D_{2s}$ and the convex chain $p_1, p_2, \ldots, p_s$. The curvature of the boundary of $\mathcal D_1$ is exaggerated in the zoom for the sake of clarity.}
  \label{fig:complement-one-even-cycle}
\end{figure}

  Now, we define the disks $\mathcal D_4, \mathcal D_6, \ldots, \mathcal D_{2s-2}$.
  For each $i \in \{4,6,\ldots,2s-2\}$, let $\mathcal D_i$ be the unique disk with the same radius as $\mathcal D_2$ and such that the boundary of $\mathcal D_i$ crosses $p_{i/2}$ and is below its tangent $\tau_{i/2}$ at this point which has the direction of $\ell(p_{i/2-1},p_{i/2+1})$.
  
  It should be observed that the only disk with even index $i$ which contains $p_{i/2}$ is $\mathcal D_i$.
  We can further choose the convex chain $\{p_i\}_{i \in [s]}$ such that one co-tangent $\tau_{i,i+1}$ to $\mathcal D_{2i}$ and $\mathcal D_{2i+2}$ has a slope between the slopes of $\tau_i$ and $\tau_{i+1}$.
  Finally we define the disks $\mathcal D_3, \mathcal D_5, \ldots, \mathcal D_{2s-1}$.
  For each $i \in \{3,5,\ldots,2s-1\}$, let $\mathcal D_i$ be tangent to $\tau_{i,i+1}$ at the point of $x$-coordinate the mean between the $x$-coordinates of $p_{\frac{i-1}{2}}$ and $p_{\frac{i+1}{2}}$.
  Moreover, $\mathcal D_i$ is above $\tau_{i,i+1}$ and has a radius sufficiently large to intersect every disk with even index which are not $\mathcal D_{i-1}$ and $\mathcal D_{i+1}$.
  It is easy to see that the disks $\mathcal D_i$ with even index (resp. odd index) form a clique.
  By construction, the disk $\mathcal D_i$ with odd index greater than 3 intersects every disk with even index except $\mathcal D_{i-1}$ and $\mathcal D_{i+1}$ since $\mathcal D_i$ is on the other side of $\tau_{i,i+1}$ than those two disks.
  As the line $\tau_{i,i+1}$ intersects every other disk with even index, there is a sufficiently large radius so that $\mathcal D_i$ does so, too.
  The particular case of $\mathcal D_1$ has been settled at the beginning of the construction.
  This disk avoids $\mathcal D_2$ and $\mathcal D_{2s}$ and contains $p_2, p_3, \ldots, p_{s-1}$, so intersects all the other disks with even index.

  We now explain how to \emph{stack} even cycles. 
  We make the distance $\varepsilon$ between the center of $\mathcal D_2$ and $\mathcal D_{2s}$ a thousandth of their common radius.
  Note that this distance does not depend on the value of $s$.
  We identify the small region (point) where the disk $\mathcal D_1$ intersects with the disks of even index, between two different complements of cycles.
  We then rotate from this point one representation by a small angle (see Figure~\ref{fig:even-cycles-complement} for multiple complements of even cycles stacked).

\begin{figure}[h!]
      \centering
      \begin{tikzpicture}[
          dot/.style={fill,circle,inner sep=-0.01cm},
          vert/.style={draw, very thin, fill=red, fill opacity=0.2},
          verta/.style={draw, very thin, fill=blue, fill opacity=0.2},
          extended line/.style={shorten >=-#1,shorten <=-#1},
          extended line/.default=1cm,
          one end extended/.style={shorten >=-#1},
          one end extended/.default=1cm,
        ]
        \def\s{1}
        \def\hs{-4}
        \def\he{4}
        \def\ve{3}
        \coordinate (c1) at (0,1) ;
        \draw[verta] (c1) circle (1) ;
        \coordinate (c2) at (0,-1) ;
        \draw[vert] (c2) circle (1) ;

        \draw[very thin] (\hs,0) -- (\he,0) ;
        \fill[blue,fill opacity=0.2] (\hs,0) -- (\he,0) -- (\he,\ve) -- (\hs,\ve) -- cycle ;

        \foreach \i in {-20,-10,10,20}{
        \begin{scope}[rotate=\i]
        \coordinate (c1) at (0,1) ;
        \draw[verta] (c1) circle (1) ;
        \coordinate (c2) at (0,-1) ;
        \draw[vert] (c2) circle (1) ;
        \draw[very thin] (\hs,0) -- (\he,0) ;

        \fill[blue,opacity=0.2] (\hs,0) -- (\he,0) -- (\he,\ve) -- (\hs,\ve) -- cycle ;
        \end{scope}
        }

        \node at (0,1) {$\mathcal D_1$} ;
        \node at (0,2.4) {$\mathcal D_{2i+1}$} ;
        \node at (0,-1) {$\mathcal D_{2i}$} ;
      \end{tikzpicture}
      \caption{A disk realization of the complement of the disjoint union of an arbitrary number of even cycles.}
      \label{fig:even-cycles-complement}
\end{figure}
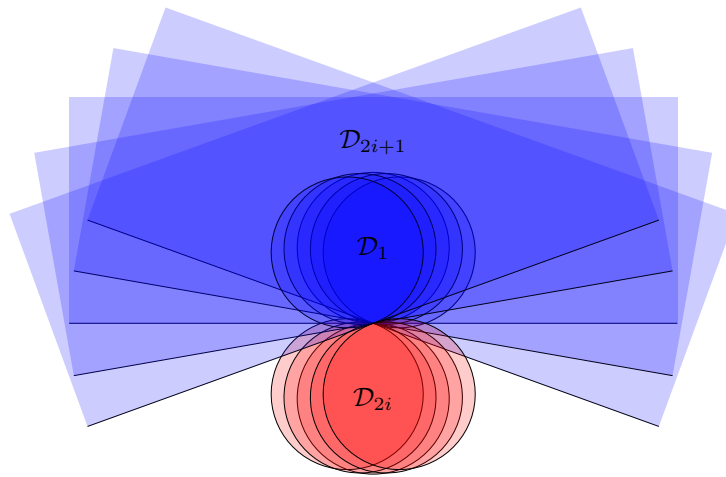

The reason why there are indeed all the edges between two complements of cycles is intuitive and depicted in Figure~\ref{fig:clique-minus-matching} and more specifically Figure~\ref{fig:cmm-lines}.
We superimpose all the complements of even cycles in a way that the maximum rotation angle between two complements of cycles is small (see for instance Figure~\ref{fig:even-cycles-odd-cycle-complement}).

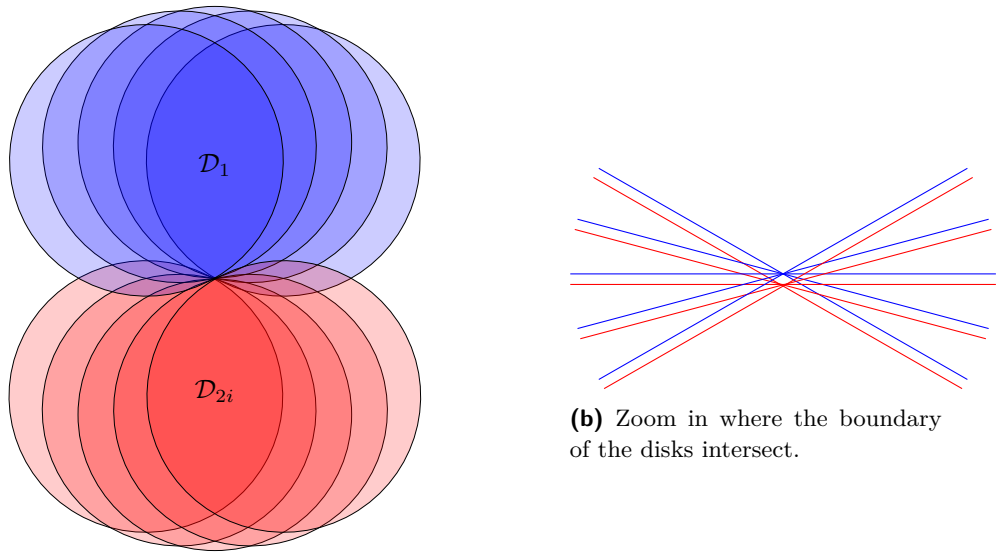
\begin{figure}[h!]
  \centering
  \begin{minipage}{0.55\textwidth}
    \centering
    \begin{tikzpicture}
      [
        scale = 1.5,
          dot/.style={fill,circle,inner sep=-0.01cm},
          vert/.style={draw, very thin, fill=red, fill opacity=0.2},
          verta/.style={draw, very thin, fill=blue, fill opacity=0.2},
          vertb/.style={draw, very thin, fill=green, fill opacity=0.2},
      ]
    \def\r{1.2}
    \def\n{5}
    \foreach \h in {1,...,\n}{
      \pgfmathsetmacro{\i}{15 * (\h - \n / 2 - 0.5)}
      \pgfmathsetmacro{\j}{25 * (\h-1)}
      \draw[rotate=\i,verta] (0,0) arc (-90:270:\r) ;
      \draw[rotate=\i,vert] (0,-0.01) arc (90:-270:\r) ;
    }
     \node at (0,1) {$\mathcal D_1$} ;
    \node at (0,-1) {$\mathcal D_{2i}$} ;
  \end{tikzpicture}
  \subcaption{The only potential non-edges are between two disks represented almost tangent.}
  \label{fig:cmm-circles}
  \end{minipage}
  \qquad
  \begin{minipage}{0.35\textwidth}
      \centering
      \begin{tikzpicture}
        [
          scale = 1.4,
          dot/.style={fill,circle,inner sep=-0.01cm},
          vert/.style={draw, very thin, fill=red, fill opacity=0.2},
          verta/.style={draw, very thin, fill=blue, fill opacity=0.2},
          vertb/.style={draw, very thin, fill=green, fill opacity=0.2},
        ]
        \def\n{5}
        \foreach \h in {1,...,\n}{
          \pgfmathsetmacro{\i}{15 * (\h - \n / 2 - 0.5)}
          \pgfmathsetmacro{\j}{25 * (\h-1)}
          \draw[rotate=\i,blue] (0,0) -- (-2,0) ;
          \draw[rotate=\i,blue] (0,0) -- (2,0) ;
          \draw[rotate=\i,red] (0,-0.1) -- (-2,-0.1) ;
          \draw[rotate=\i,red] (0,-0.1) -- (2,-0.1) ;
        }
  \end{tikzpicture}
  \subcaption{Zoom in where the boundary of the disks intersect.}
  \label{fig:cmm-lines}
  \end{minipage}
   \caption{Zoom in where the disk $\mathcal D_1$ of the several complements of even cycles intersects all the $\mathcal D_{2i}$ of the other cycles.}
  \label{fig:clique-minus-matching}
\end{figure}

Finally, we need to add one disjoint odd cycle in the complement.
There is a nice representation of a complement of an odd cycle by unit disks in the paper of Atminas and Zamaraev \cite{Atminas16} (see Figure~\ref{fig:atminas-zamaraev}).
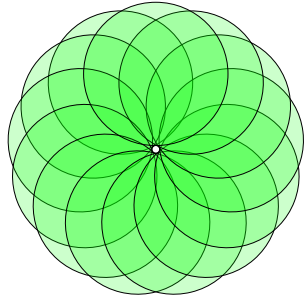
\begin{figure}[h!]
      \centering
      \begin{tikzpicture}[
          dot/.style={fill,circle,inner sep=-0.01cm},
          vert/.style={draw, very thin, fill=red, fill opacity=0.2},
          verta/.style={draw, very thin, fill=blue, fill opacity=0.2},
          vertb/.style={draw, very thin, fill=green, fill opacity=0.2},
        ]
        \def\s{13}
        \foreach \i in {1,...,\s}{
        \begin{scope}[rotate=360 * \i / 13]
        \coordinate (d\i) at (0,1) ;
        \draw[vertb] (d\i) circle (0.95) ;
        \end{scope}
        }
        
      \end{tikzpicture}
      \caption{A disk realization of the complement of an odd cycle with unit disks as described by Atminas and Zamaraev \cite{Atminas16}. Unfortunately, we cannot use this representation.}
      \label{fig:atminas-zamaraev}
\end{figure}

We will use a different and non-unit representation for the next step to work.
Let $2s+1$ be the length of the cycle.
We use a similar construction as for the complement of an even cycle.
We denote the disks $\mathcal D'_1, \mathcal D'_2, \ldots, \mathcal D'_{2s+1}$.
The difference is that we separate $\mathcal D'_1$ away from $\mathcal D'_2$ but not from $\mathcal D'_{2s}$.
Then, we represent all the disks with odd index but $\mathcal D'_{2s+1}$ as before.
The disk $\mathcal D'_{2s+1}$ is chosen as being cotangent to $\mathcal D'_1$ and $\mathcal D'_{2s}$ and to the left of them.
Then we very slighlty move $\mathcal D'_{2s+1}$ to the left so that it does not intersect those two disks anymore.
The disk $\mathcal D'_{2s}$ have the rightmost center among the disks with even index.
Therefore $\mathcal D'_{2s+1}$ still intersects all the other disks of even index.

Moreover, the disks with even index form a clique and the disks with odd index form a clique minus an edge between the vertex $1$ and the vertex $2s+1$.
Hence, the intersection graph of those disks is indeed the complement of $C_{2s+1}$ (see Figure~\ref{fig:odd-cycle-complement}).

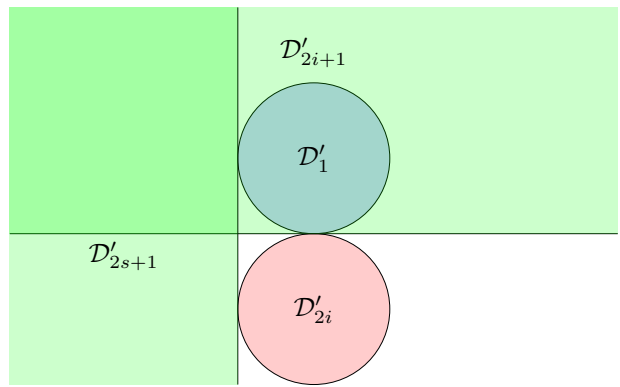
\begin{figure}[h!]
      \centering
      \begin{tikzpicture}[
          dot/.style={fill,circle,inner sep=-0.01cm},
          vert/.style={draw, very thin, fill=red, fill opacity=0.2},
          verta/.style={draw, very thin, fill=blue, fill opacity=0.2},
          vertb/.style={draw, very thin, fill=green, fill opacity=0.2},
        ]
        \def\s{1}
        \def\hs{-4}
        \def\he{4}
        \def\ve{3}
        \def\vs{-2}
        \coordinate (c1) at (0,1) ;
        \draw[verta] (c1) circle (1) ;
        \coordinate (c2) at (0,-1) ;
        \draw[vert] (c2) circle (1) ;

        \draw[very thin] (\hs,0) -- (\he,0) ;
        \fill[green,fill opacity=0.2] (\hs,0) -- (\he,0) -- (\he,\ve) -- (\hs,\ve) -- cycle ;

        \draw[very thin] (-1,\vs) -- (-1,\ve) ;
        \fill[green,fill opacity=0.2] (-1,\vs) -- (-1,\ve) -- (\hs,\ve) -- (\hs,\vs) -- cycle ;
        
        \node at (0,2.4) {$\mathcal D'_{2i+1}$} ;
        \node at (0,-1) {$\mathcal D'_{2i}$} ;
        \node at (-2.5,-0.3) {$\mathcal D'_{2s+1}$} ;
        \node at (0,1) {$\mathcal D'_1$} ;
      \end{tikzpicture}
      \caption{A disk realization of the complement of an odd cycle of length $2s+1$.}
      \label{fig:odd-cycle-complement}
\end{figure}

This representation of $\overline{C_{2s+1}}$ can now be put on top of complements of even cycles.
We identify the small region (point) where the disk $\mathcal D_1$ intersects the disks of even index (in complements of even cycles) with the small region (point) where the disk $\mathcal D'_1$ intersects the disks of even index (in the one complement of odd cycle).
We make the disk $\mathcal D'_1$ significantly smaller than $\mathcal D_1$ and rotate the representation of $\overline{C_{2s+1}}$ by a sizable angle, say 60 degrees (see Figure~\ref{fig:even-cycles-odd-cycle-complement}). 

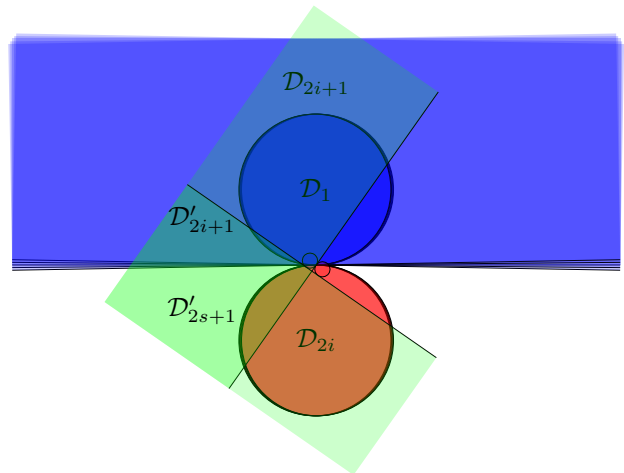
\begin{figure}[h!]
      \centering
      \begin{tikzpicture}[
          dot/.style={fill,circle,inner sep=-0.01cm},
          vert/.style={draw, very thin, fill=red, fill opacity=0.2},
          verta/.style={draw, very thin, fill=blue, fill opacity=0.2},
          vertb/.style={draw, very thin, fill=green, fill opacity=0.2},
        ]        
         \def\s{1}
        \def\hs{-4}
        \def\he{4}
        \def\ve{3}
        \coordinate (c1) at (0,1) ;
        \draw[verta] (c1) circle (1) ;
        \coordinate (c2) at (0,-1) ;
        \draw[vert] (c2) circle (1) ;

        \draw[very thin] (\hs,0) -- (\he,0) ;
        \fill[blue,fill opacity=0.2] (\hs,0) -- (\he,0) -- (\he,\ve) -- (\hs,\ve) -- cycle ;

        \foreach \i in {-1,-0.5,0.5,1}{
        \begin{scope}[rotate=\i]
        \coordinate (c1) at (0,1) ;
        \draw[verta] (c1) circle (1) ;
        \coordinate (c2) at (0,-1) ;
        \draw[vert] (c2) circle (1) ;
        \draw[very thin] (\hs,0) -- (\he,0) ;

        \fill[blue,opacity=0.2] (\hs,0) -- (\he,0) -- (\he,\ve) -- (\hs,\ve) -- cycle ;
        \end{scope}
        }

        \node at (0,1) {$\mathcal D_1$} ;
        \node at (0,2.4) {$\mathcal D_{2i+1}$} ;
        \node at (0,-1) {$\mathcal D_{2i}$} ;

        \def\t{0.1}
        \def\nve{2}
        \def\nhs{-2}
        \def\nhe{2.8}
        
        \begin{scope}[rotate=55]
        \coordinate (d1) at (0,\t) ;
        \draw[verta] (d1) circle (\t) ;
        \coordinate (d2) at (0,-\t) ;
        \draw[vert] (d2) circle (\t) ;

        \draw[very thin] (\nhs,0) -- (\nhe,0) ;
        \fill[green,fill opacity=0.2] (\nhs,0) -- (\nhe,0) -- (\nhe,\nve) -- (\nhs,\nve) -- cycle ;

        \draw[very thin] (-\t,-\nve) -- (-\t,\nve) ;
        \fill[green,fill opacity=0.2] (-\t,-\nve) -- (-\t,\nve) -- (\nhs,\nve) -- (\nhs,-\nve) -- cycle ;
        \end{scope}

        \node at (-1.5,0.6) {$\mathcal D'_{2i+1}$} ;
        \node at (-1.5,-0.6) {$\mathcal D'_{2s+1}$} ;

      \end{tikzpicture}
      \caption{Placing the complement of odd cycle on top of the complements of even cycles.}
      \label{fig:even-cycles-odd-cycle-complement}
\end{figure}

It is easy to see that the disks of the complement of the odd cycle intersect all the disks of the complements of even cycles.
A good sanity check is to observe why we cannot stack representations of complements of odd cycles, with the same rotation scheme.
In Figure~\ref{fig:sanity-check}, the rotation of two representations of the complement of an odd cycle leaves disks $\mathcal D'_1$ and $\mathcal D''_{2s'+1}$ far apart when they should intersect.
\begin{figure}[h!]
      \centering
      \begin{tikzpicture}[
          dot/.style={fill,circle,inner sep=-0.01cm},
          vert/.style={draw, very thin, fill=red, fill opacity=0.2},
          verta/.style={draw, very thin, fill=blue, fill opacity=0.2},
          vertb/.style={draw, very thin, fill=green, fill opacity=0.2},
        ]
        \def\s{1}
        \def\hs{-4}
        \def\he{4}
        \def\ve{3}
        \def\vs{-2}
        \foreach \i in {0,20}{
        \begin{scope}[rotate=\i]
        \coordinate (c1) at (0,1) ;
        \draw[verta] (c1) circle (1) ;
        \coordinate (c2) at (0,-1) ;
        \draw[vert] (c2) circle (1) ;

        \draw[very thin] (\hs,0) -- (\he,0) ;
        \fill[green,fill opacity=0.2] (\hs,0) -- (\he,0) -- (\he,\ve) -- (\hs,\ve) -- cycle ;

        \draw[very thin] (-1,\vs) -- (-1,\ve) ;
        \fill[green,fill opacity=0.2] (-1,\vs) -- (-1,\ve) -- (\hs,\ve) -- (\hs,\vs) -- cycle ;
        \end{scope}
        }
        
        \node at (0,2.4) {$\mathcal D'_{2i+1}$} ;
        \node at (0,-1) {$\mathcal D'_{2i}$} ;
        \node at (-2.5,-0.3) {$\mathcal D'_{2s+1}$} ;
        \node at (0,1) {$\mathcal D'_1$} ;

        \node at (-2.5,-1.3) {$\mathcal D''_{2s'+1}$} ;
      \end{tikzpicture}
      \caption{Sanity check: trying to stack the complements of two odd cycles fails.
      The disks $\mathcal D'_1$ and $\mathcal D''_{2s'+1}$ do not intersect.}
      \label{fig:sanity-check}
\end{figure}
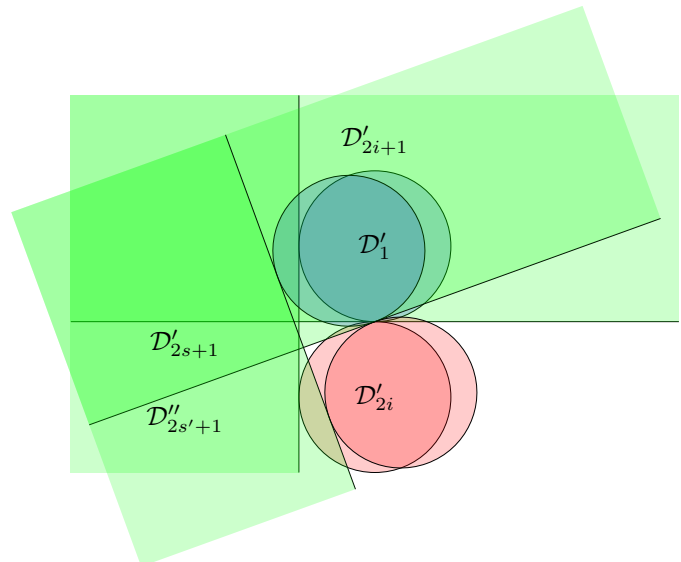

\end{proof}

Theorem~\ref{thm:main-structural-non-disk} and Theorem~\ref{thm:coEvenCycles}, together with the fact that disk graphs are closed by taking induced subgraphs prove Theorem~\ref{thm:main-structural}.

\section{Algorithmic consequences}\label{sec:algorithms}

Now we show how to use the structural results from Section \ref{sec:structural} to obtain algorithms for \cli in disk graphs.
A clique in a graph $G$ is an independent set in $\overline{G}$. 
So, leveraging the result from Theorem \ref{thm:main-structural}, we will focus on solving \mis in graphs without two vertex-disjoint odd cycles as an induced subgraph.
 
\subsection{QPTAS}

The odd cycle packing number $\ocp(H)$ of a graph $H$ is the maximum number of vertex-disjoint odd cycles in $H$.
Unfortunately, the condition that $\overline{G}$ does not contain two vertex-disjoint odd cycles as an induced subgraph is not quite the same as saying that the odd cycle packing number of $\overline{G}$ is 1.
Otherwise, we would immediately get a PTAS by the following result of Bock et al.~\cite{Bock14}.
\begin{theorem}[Bock et al.~\cite{Bock14}]\label{thm:bock-ptas}
For every fixed $\varepsilon >0$ there is a polynomial $(1+\varepsilon)$-approximation algorithm for \mis for graphs $H$ with $n$ vertices and  $\ocp(H) = o(n / \log n)$.
\end{theorem}
The algorithm by Bock et al. works in polynomial time if $\ocp(H) = o(n / \log n)$, but it does not need the odd cycle packing explicitly given as an input. This is important, since finding a maximum odd cycle packing is NP-hard \cite{DBLP:conf/stoc/KawarabayashiR10}.
We start by proving a structural lemma, which spares us having to determine the odd cycle packing number.

\begin{lemma}\label{lem:bigdeg}
Let $H$ be a graph with $n$ vertices, whose complement is a disk graph.
If $\ocp(H) > n/ \log^2 n$, then $H$ has a vertex of degree at least $n / \log ^4 n$.
\end{lemma}
\begin{proof}
  Consider a maximum odd cycle packing $\mathcal{C}$.
  By assumption, it contains more than $n/\log^2 n$ vertex-disjoint cycles. 
By the pigeonhole principle, there must be a cycle $C \in \mathcal{C}$ of size at most $\log^2 n$.
Now, by Theorem \ref{thm:main-structural-non-disk}, $H$ has no two vertex-disjoint odd cycles with no edges between them. Therefore there must be an edge from $C$ to every other cycle of $\mathcal{C}$, there are at least $n / \log^2 n$ such edges. Let $v$ be a vertex of $C$ with the maximum number of edges to other cycles in $\mathcal{C}$, by the pigeonhole principle its degree is at least $n / \log^4 n$.
\end{proof}

Now we are ready to construct a QPTAS for \cli in disk graphs.

\begin{theorem}\label{thm:qptas}
For any $\varepsilon > 0$, \cli can be $(1+\varepsilon)$-approximated in time $2^{O(\log^5 n)}$, when the input is a disk graph with $n$ vertices.
\end{theorem}

\begin{proof}
Let $G$ be the input disk graph and let $\overline{G}$ be its complement, we want to find a $(1+\varepsilon)$-approximation for \mis in $\overline{G}$. We consider two cases.
If $\overline{G}$ has no vertex of degree at least $n / \log ^4 n$, then, by Lemma \ref{lem:bigdeg}, we know that $\ocp(\overline{G}) \leqslant n / \log^2 n = o(n / \log n)$. In this case we run the PTAS of Bock et al. and we are done.

In the other case, $\overline{G}$ has a vertex $v$ of degree at least $n / \log ^4 n$ (note that it may still be the case that $\ocp(\overline{G}) = o(n / \log n)$). We branch on $v$: either we include $v$ in our solution and remove it and all its neighbors, or we discard $v$. The complexity of this step is described by the recursion $F(n) \leqslant F(n-1) + F(n- n / \log^4 n)$ and solving it gives us the desired running time. Note that this step is exact, i.e., we do not lose any solutions.
\end{proof}

\subsection{Subexponential algorithm}

Now we will show how our structural result can be used to construct a subexponential algorithm for \cli in disk graphs.
The \emph{odd girth} of a graph is the size of a shortest odd cycle.
An \emph{odd cycle cover} is a subset of vertices whose deletion makes the graph bipartite.
We will use a result by Györi et al. \cite{Gyori97}, which says that graphs with large odd girth have small odd cycle cover.
In that sense, it can be seen as relativizing the fact that odd cycles do not have the Erd\H{o}s-P\'osa property. Bock et al. \cite{Bock14} turned the non-constructive proof into a polynomial-time algorithm.
\begin{theorem}[Györi et al. \cite{Gyori97}, Bock et al. \cite{Bock14}]\label{thm:occ}
Let $H$ be a graph with $n$ vertices and no odd cycle shorter than $\delta n$ ($\delta$ may be a function of $n$).
Then there is an odd cycle cover $X$ of size  at most $(48/\delta) \ln (5/\delta)$
Moreover, $X$ can be found in polynomial time.
\end{theorem}

Let us start with showing three variants of an algorithm.
\begin{theorem} \label{thm:subexp}
  Let $G$ be a disk graph with $n$ vertices. Let $\Delta$ be the maximum degree of $\overline{G}$ and $c$ the odd girth of $\overline{G}$ (they may be functions of $n$).
  \cli has a branching or can be solved, up to a polynomial factor, in time:\\
\begin{enumerate*}[label=(\roman*),itemjoin={\quad}]
\item $2^{\tilde{O}(n/\Delta)}$ (branching), \label{case:subexp-delta}
\item $2^{\tilde{O}(n/c)}$ (solved),		\label{case:subexp-oddgirth}
\item $2^{{O}(c  \Delta)}$ (solved). \label{case:subexp-both}
\end{enumerate*}
\end{theorem}
\begin{proof}
Let $G$ be the input disk graph and let $\overline{G}$ be its complement, we look for a maximum independent set in $\overline{G}$. 

To prove \ref{case:subexp-delta}, consider a vertex $v$ of degree $\Delta$ in $\overline{G}$. We branch on $v$: either we include $v$ in our solution and remove $N[v]$, or discard $v$. The complexity is described by the recursion $F(n) \leqslant F(n-1) + F(n- (\Delta+1))$ and solving it gives \ref{case:subexp-delta}.
Observe that this does not give an algorithm running in time $2^{\tilde{O}(n/\Delta)}$ since the maximum degree might drop.
Therefore, we will do this branching as long as it is \emph{good enough} and then finish with the algorithms corresponding to \ref{case:subexp-oddgirth} and \ref{case:subexp-both}.

For \ref{case:subexp-oddgirth} and \ref{case:subexp-both}, let $C$ be the cycle of length $c$, it clearly can be found in polynomial time. By application of Theorem \ref{thm:occ} with $\delta = c/n$, we find an odd cycle cover $X$ in $\overline{G}$ of size $\tilde{O}(n/c)$ in polynomial time (see for instance \cite{AlonYZ97}). Next we exhaustively guess in time $2^{\tilde{O}(n/c)}$ the intersection $I$ of an optimum solution with $X$ and finish by finding a maximum independent set in the bipartite graph $\overline{G}-(X \cup N(I))$, which can be done in polynomial time. The total complexity of this case is $2^{\tilde{O}(n/c)}$, which shows \ref{case:subexp-oddgirth}.

Finally, observe that the graph $\overline{G} - N[C]$ is bipartite, since otherwise $\overline{G}$ contains two vertex-disjoint odd cycles with no edges between them.
Moreover, since every vertex in $\overline{G}$ has degree at most $\Delta$, it holds that $|N[C]| \leqslant c  (\Delta-1) \leqslant c  \Delta$.
Indeed, a vertex of $C$ can only have $c  (\Delta-2)$ neighbors outside $C$. 
We can proceed as in the previous step: we exhaustively guess the intersection of the optimal solution with $N[C]$ and finish by finding the maximum independent set in a bipartite graph (a subgraph of $\overline{G}-N[C]$), which can be done in total time $2^{O(c  \Delta)}$, which shows \ref{case:subexp-both}.
\end{proof}

Now we show how the structure of $G$ affect the bounds in Theorem \ref{thm:subexp}.

\begin{corollary}\label{cor:subexp}
Let $G$ be a disk graph with $n$ vertices. \cli can be solved in time:
\begin{compactenum}[(a)]
\item $2^{\tilde{O}(n^{2/3})}$,
\item $2^{\tilde{O}(\sqrt{n})}$ if the maximum degree of $\overline{G}$ is constant,
\item polynomial, if both the maximum degree and the odd girth of $\overline{G}$ are constant.
\end{compactenum}
\end{corollary}
\begin{proof}
  $\Delta$ and $c$ can be computed in polynomial time.
  Therefore, knowing what is faster among cases \ref{case:subexp-delta}, \ref{case:subexp-oddgirth}, and \ref{case:subexp-both} is tractable.
  For case (a), while there is a vertex of degree at least $n^{1/3}$, we branch on it.
  When this process stops, we do what is more advantageous between cases \ref{case:subexp-oddgirth} and \ref{case:subexp-both}.
  Note that $\min (n/\Delta, n/c, c \Delta) \leqslant n^{2/3}$ (the equality is met for $\Delta = c = n^{1/3}$).
  For case (b), we do what is best between cases \ref{case:subexp-oddgirth} and \ref{case:subexp-both}.
  Note that $\min (n/c, c) \leqslant \sqrt{n}$ (the equality is met for $c = \sqrt{n})$.
Finally, case (c) follows directly from case \ref{case:subexp-both} in Theorem \ref{thm:subexp}.
\end{proof}
Observe that case (b) is typically the hardest one for \cli.
Moreover, the win-win strategy of Corollary \ref{cor:subexp} can be directly applied to solve \wcli, as finding a maximum weighted independent set in a bipartite graph is still polynomial-time solvable.
On the other hand, this approach cannot be easily adapted to obtain a subexponential algorithm for \textsc{Clique Partition} (even \textsc{Clique $p$-Partition} with constant $p$), since \textsc{List Coloring} (even \textsc{List $3$-Coloring}) has no subexponential algorithm for bipartite graphs, unless the ETH fails~(see \cite{precoloring}, the bound can be obtained if we start reduction from a sparse instance of {\textsc{1-in-3-Sat} instead of {\textsc{Planar 1-in-3-Sat}).

\section{Other intersection graphs and limits}\label{sec:gen&lim}
In this section, we discuss the impossibility of generalizing our results to related classes of intersection graphs.

\subsection{Filled ellipses and filled triangles}

A natural generalization of a disk is an \emph{elliptical disk}, also called \emph{filled ellipse}, i.e., an ellipse plus its interior.
The simplest convex set with non empty interior is a filled triangle (a triangle plus its interior).  
We show that our approach developed in the two previous sections, and actually every approach, is bound to fail for filled ellipses and filled triangles.

APX-hardness was shown for \cli in the intersection graphs of (\emph{non-filled}) ellipses and triangles by Amb\"uhl and Wagner~\cite{Ambuhl05}.
Their reduction also implies that there is no subexponential algorithm for this problem, unless the ETH fails.
Moreover, they claim that their hardness result extends to filled ellipses since \emph{``intersection graphs of ellipses without interior are also intersection graphs of filled ellipses''}.
Unfortunately, as we show below, this claim is incorrect.

\begin{theorem}\label{thm:counterexample}
There is a graph $G$ which has an intersection representation with ellipses without their interior, but has no intersection representation with convex sets.
\end{theorem}

\begin{proof}
The argument is similar to the one used by Brimkov et al. \cite{DBLP:journals/corr/BrimkovJKKPRST14}, which was in turn inspired by the construction by Kratochv\'il and Matou\v{s}ek \cite{DBLP:journals/jct/KratochvilM94}. 
Consider the graph $G$ in Figure~\ref{fig:counterexample2} (containing what we will henceforth call \textit{black}, \textit{gray} and \textit{white} vertices), and observe that $c$ and $d$ are two non-adjacent vertices with the same neighborhoods.

\begin{figure}[ht]
\centering
\begin{center}
\tiny
\begin{tikzpicture}[scale=0.72]
\tikzstyle{every node}=[draw, shape = circle]

\node (o1) at (0,3) {};
\node (o2) at (1,2.6) {};
\node (o3) at (2,2) {};
\node (o4) at (2.6,1) {};
\node (o5) at (3,0) {};
\node (o6) at (2.6,-1) {};
\node (o7) at (2,-2) {};
\node (o8) at (1,-2.6) {};
\node (o9) at (0,-3) {};
\node (o10) at (-1,-2.6) {};
\node (o11) at (-2,-2) {};
\node (o12) at (-2.6,-1) {};
\node (o13) at (-3,0) {};
\node (o14) at (-2.6,1) {};
\node (o15) at (-2,2) {};
\node (o16) at (-1,2.6) {};

\node[fill = black, inner sep=-0.19cm] (a) at (0,0) {\color{white}{\large $c$}}; 
\node[fill = black, inner sep=-0.24cm] (ap) at (1,-0.3) {\color{white}{\large $d$}}; 
\node[fill = black, inner sep=-0.23cm] (b) at (-1,-0.5) {\color{white}{\large $b$}}; 
\node[fill = black, inner sep=-0.20cm] (c) at (0,1) {\color{white}{\large $a$}}; 

\node[fill = gray] (a3) at (1.2,1.2) { };
\node[fill = gray] (a7) at (1.2,-1.2) { };
\node[fill = gray] (a11) at (-1.2,-1.2) { };
\node[fill = gray] (a15) at (-1.2,1.2) { };
\node[fill = gray] (b1) at (-0.5,1.5) { };
\node[fill = gray] (b9) at (-0.5,-1.5) { };
\node[fill = gray] (c5) at (1.5,0.5) { };
\node[fill = gray] (c13) at (-1.5,0.5) { };
\draw (o1)--(o2)--(o3)--(o4)--(o5)--(o6)--(o7)--(o8)--(o9)--(o10)--(o11)--(o12)--(o13)--(o14)--(o15)--(o16)--(o1);
\draw (a) -- (a3) -- (o3);
\draw (a) -- (a7) -- (o7);
\draw (a) -- (a11) -- (o11);
\draw (a) -- (a15) -- (o15);

\draw (ap) -- (a3) -- (o3);
\draw (ap) -- (a7) -- (o7);
\draw (ap) -- (a11) -- (o11);
\draw (ap) -- (a15) -- (o15);

\draw (b) -- (b1) -- (o1);
\draw (b) -- (b9) -- (o9);
\draw (c) -- (c5) -- (o5);
\draw (c) -- (c13) -- (o13);
\draw (a) -- (b) -- (c) -- (a);
\draw (ap) -- (b) -- (c) -- (ap);
\end{tikzpicture}
\hskip 14 pt
\begin{tikzpicture}[scale=0.45]
\draw (0,0) ellipse (2.5 and 2.5);
\draw (0,0) ellipse (2.2 and 2.2);

\draw (0,0) ellipse (3.1 and 0.9);
\draw (0,0) ellipse (0.9 and 3.1);

\node at (0.7,3) {\large $a$};
\node at (-1.87,0.0) {\large $c$};
\node at (2.5,1.6) {\large $d$};
\node at (3,0.8) {\large $b$};

\draw (-3.75,4.5) ellipse (2.5 and 0.75);
\draw (3.75,4.5) ellipse (2.5 and 0.75);
\draw (3.75,-4.5) ellipse (2.5 and 0.75);
\draw (-3.75,-4.5) ellipse (2.5 and 0.75);
\draw (0,6.5) ellipse (1.75 and 0.75);
\draw (0,-6.5) ellipse (1.75 and 0.75);
\draw (-6.5,0) ellipse (0.75 and 2.25);
\draw (6.5,0) ellipse (0.75 and 2.25);
\draw (1.5, 5.5) ellipse (0.5 and 1);
\draw (1.5, -5.5) ellipse (0.5 and 1);
\draw (-1.5, 5.5) ellipse (0.5 and 1);
\draw (-1.5, -5.5) ellipse (0.5 and 1);
\draw (6, 3) ellipse (0.5 and 1.5);
\draw (6, -3) ellipse (0.5 and 1.5);
\draw (-6, 3) ellipse (0.5 and 1.5);
\draw (-6, -3) ellipse (0.5 and 1.5);
\draw (0, 4.5) ellipse (0.25 and 1.75);
\draw (0, -4.5) ellipse (0.25 and 1.75);
\draw (4.5, 0) ellipse (1.75 and 0.25);
\draw (-4.5, 0) ellipse (1.75 and 0.25);
\draw (1.75, 2.7) ellipse (0.15 and 1.7);
\draw (-1.75, 2.7) ellipse (0.15 and 1.7);
\draw (1.75, -2.7) ellipse (0.15 and 1.7);
\draw (-1.75, -2.7) ellipse (0.15 and 1.7);
\end{tikzpicture}

\end{center}
\caption{A graph (left), which has a representation with empty ellipses (right) but no representation with convex sets.}
\label{fig:counterexample2}
\end{figure}
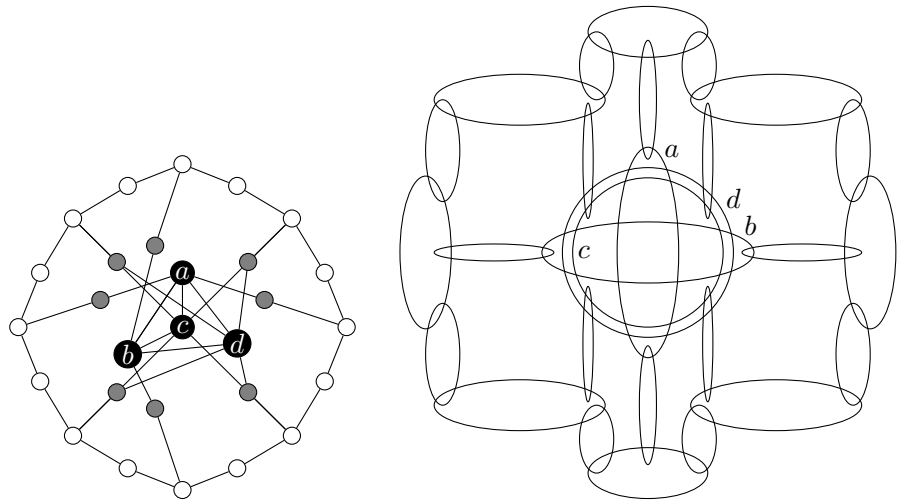

Suppose $G$ can be represented by intersecting convex sets. For a vertex $v$, let $R_v$ be the convex set representing $v$.
The union of representatives of the white vertices contains a closed Jordan curve, that we will call the outer circle. Let us choose the outer circle in such a way that it intersects the representatives of all gray vertices. It divides the plane into two faces -- an interior and an exterior.

The outer circle cannot be crossed by the representative of any black vertex. Moreover, as black vertices form a connected subgraph, they have to be represented in the same face $F$ (with respect to the outer circle). Thus, along this circle the representatives of gray vertices appear in a prescribed ordering (note that they form an independent set). This implies the ordering in which some part of representatives of the black vertices occur.

First, observe that the representatives of the gray neighbors of $a,b$, and $c$ intersect the outer circle in the following ordering: $a_1,c_1,b_1,c_2,a_2,c_3,b_2,c_4$ (where each $z_i$ for $z \in \{a,b,c\}$ is a distinct gray neighbor of $z$).

Clearly, each gray neighbor of $a$ must intersect $R_a$ outside $R_a \cap (R_b \cup R_c)$, each gray neighbor of $b$ must intersect $R_b$ outside $R_b \cap (R_a \cup R_c)$, and each gray neighbor of $c$ must intersect $R_c$ outside $R_c \cap (R_a \cup R_b)$.  Thus, some parts of $R_a$, $R_b$, and $R_c$ are exposed (i.e., outside the intersection with the union of representatives of remaining two vertices) in the ordering: $a,c,b,c,a,c,b$, as we move along the boundary of $R_a \cup R_b \cup R_c$. Note that this implies that $R_a \cap R_b \cap R_c \neq \emptyset$, since all sets are convex.

For any $z \in \{a,b,c\}$ and any $i$, the set $R_{z_i}$ contains a segment $s(z_i)$, whose one end is on the boundary of $R_z$ and the other end is on the outer circle (recall that all representatives are convex). For $z \in \{a,b\}$ and $i \in \{1,2\}$, by $s'(z_i)$ we denote the segment joining the endpoint of $s(z_i)$ on the boundary of $R_z$ to the closest point in $R_z \cap R_c$.
Now we observe that the set $\bigcup_{z \in \{a,b\}, i \in \{1,2\}} s(z_i) \cup s'(z_i)$ partitions $F \setminus R_c$ into four disjoint regions $Q_1,Q_2,Q_3,Q_4$. Let $Q_1$ be the region adjacent to $s(a_1)$ and $s(b_1)$, $Q_2$ be the region adjacent to $s(b_1)$ and $s(a_2)$, $Q_3$ be the region adjacent to $s(a_2)$ and $s(b_2)$, and $Q_4$ be the region adjacent to $s(b_2)$ and $s(a_1)$. Note that one of these regions may be unbounded, if $F$ is the unbounded face of the outer circle.

For every $i \in \{1,2,3,4\}$, the set  $R_{c_i} \setminus R_c$ is contained in $Q_i$.
For $i = \{1,2,3,4\}$, let $p_i$ be a point in $R_d \cap R_{c_i}$, such a point exist, since $d$ is adjacent to $c_i$.
By convexity of $R_d$, the segment $p_1p_2$ is contained in $R_d$. On the other hand, it crosses the curve $s(b_1) \cup s'(b_1)$, let $q_1$ be the intersection point. Since $R_d$ is disjoint with $R_{b_1}$, clearly $q_1 \in s'(b_1) \subseteq R_b$.
In the analogous way we define $q_2$ to be the crossing point of $p_2p_3$ and $s(a_2) \cup s'(a_2)$, $q_3$ to be the crossing point of $p_3p_4$ and $s(b_2) \cup s'(b_2)$, and $q_4$ to be the crossing point of $p_4p_1$ and $s(a_1) \cup s'(a_1)$. We observe that $q_2 \in s'(a_2) \subseteq R_a$, $q_3 \in s'(b_2) \subseteq R_b$, and $q_4 \in s'(a_1) \subseteq R_a$.
Let us consider the segment $q_1q_3$.
It must intersect  either $s(c_2) \cup R_c$ or $s(c_4) \cup R_c$.
Without loss of generality, we assume that it intersects $s(c_2) \cup R_c$.
Let $q'$ be this intersection point.
By convexity, $q' \in R_d$ and $q' \in R_b$.
If $q' \in s(c_2)$, we get the contradiction with the fact that $b$ and $c_2$ are non-adjacent.
On the other hand, if $q' \in R_c$, we get the contradiction with the fact that $d$ and $c$ are non-adjacent.

Finally, it is easy to represent $G$ with empty ellipses (see Fig. \ref{fig:counterexample2} right).
\end{proof}

This error and the confusion between filled ellipses and ellipses without their interior has propagated to other more recent papers \cite{Keller17}.
Fortunately, we show that the hardness result does hold for filled ellipses (and filled triangles) with a different reduction.
Our construction can be seen as streamlining the ideas of Ambühl and Wagner \cite{Ambuhl05}.
It is simpler and, in the case of (filled) ellipses, yields a somewhat stronger statement.

\begin{theorem}\label{thm:hardness-filled-ellipses}
  There is a constant $\alpha > 1$ such that for every $\varepsilon > 0$, \cli on the intersection graphs of filled ellipses has no $\alpha$-approximation algorithm running in subexponential time $2^{n^{1-\varepsilon}}$, unless the ETH fails, even when the ellipses have arbitrarily small eccentricity and arbitrarily close value of major axis.
\end{theorem}

This is in sharp contrast with our subexponential algorithm and with our QPTAS when the eccentricity is 0 (case of disks).
For any $\varepsilon > 0$, if the eccentricity is only allowed to be at most $\varepsilon$, a subexponential algorithm or a QPTAS are very unlikely.
This result subsumes \cite{ceroi} (where NP-hardness is shown for connected shapes contained in a disk of radius 1 and containing a concentric disk of radius $1-\varepsilon$ for arbitrarily small $\varepsilon > 0$) and corrects \cite{Ambuhl05}.
We show the same hardness for the intersection graphs of filled triangles.

\begin{theorem}\label{thm:hardness-filled-triangles}
 There is a constant $\alpha > 1$ such that for every $\varepsilon > 0$, \cli on the intersection graphs of filled triangles has no $\alpha$-approximation algorithm running in subexponential time $2^{n^{1-\varepsilon}}$, unless the ETH fails.
\end{theorem}

We first show this lower bound for \textsc{Maximum Weighted Independent Set} on the class of all the 2-subdivisions, hence the same hardness for \textsc{Maximum Weighted Clique} on all the co-2-subdivisions.
It is folklore that from the PCP of Moshkovitz and Raz \cite{Moshkovitz10}, which roughly implies that \textsc{Max 3-SAT} cannot be $7/8+\varepsilon$-approximated in subexponential time under the ETH, one can derive such inapproximability in subexponential time for many hard graph and hypergraph problems; see for instance \cite{Bonnet15}.

The following inapproximability result for \mis on bounded-degree graphs was shown by Chleb\'ik and Chleb\'ikov\'a \cite{Chlebik06}.
As their reduction is almost linear, the PCP of Moshkovitz and Raz boosts this hardness result from ruling out polynomial-time up to ruling out subexponential time $2^{n^{1-\varepsilon}}$ for any $\varepsilon > 0$. 
\begin{theorem}[\cite{Chlebik06,Moshkovitz10}]\label{thm:inapprox-mis}
There is a constant $\beta > 0$ such that \mis on graphs with $n$ vertices and maximum degree $\Delta$ cannot be $1+\beta$-approximated in time $2^{n^{1-\varepsilon}}$ for any $\varepsilon > 0$, unless the ETH fails. 
\end{theorem}
We could actually state a slightly stronger statement for the running time but will settle for this for the sake of clarity.

\begin{theorem}\label{thm:hardness-2subd}
  There is a constant $\alpha > 1$ such that for any $\varepsilon > 0$, \textsc{Maximum Independent Set} on the class of all the 2-subdivisions has no $\alpha$-approximation algorithm running in subexponential time $2^{n^{1-\varepsilon}}$, unless the ETH fails.
\end{theorem}
\begin{proof}
Let $G$ be a graph with maximum degree a constant $\Delta$, with $n$ vertices $v_1, \ldots, v_n$ and $m$ edges $e_1, \ldots, e_m$, and let $H$ be its 2-subdivision.
Recall that to form $H$, we subdivided every edge of $G$ exactly twice.
These $2m$ vertices in $V(H) \setminus V(G)$, representing edges, are called \emph{edge vertices} and are denoted by $v^+(e_1), v^-(e_1), \ldots, v^+(e_m), v^-(e_m)$, as opposed to the other vertices of $H$, which we call \emph{original vertices}.
If $e_k=v_iv_j$ is an edge of $G$, then $v^+(e_k)$ (resp. $v^-(e_k)$) has two neighbors: $v^-(e_k)$ and $v_i$ (resp. $v^+(e_k)$ and $v_j$). 

Observe that there is a maximum independent set $S$ which contains exactly one of $v^+(e_k), v^-(e_k)$ for every $k \in [m]$. Indeed, $S$ cannot contain both $v^+(e_k)$ and $v^-(e_k)$ since they are adjacent. On the other hand, if $S$ contains neither $v^+(e_k)$ nor $v^-(e_k)$, then adding $v^+(e_k)$ to $S$ and potentially removing the other neighbor of $v^+(e_k)$ which is $v_i$ (with $e_k=v_iv_j$) can only increase the size of the independent set.
Hence $S$ contains $m$ edge vertices and $s \leqslant n$ original vertices, and there is no larger independent set in $H$.

We observe that the $s$ original vertices is $S$ form an independent set in $G$.
Indeed, if  $v_iv_j=e_k \in E(G)$ and $v_i,v_j \in S$, then neither $v^+(e_k)$ nor $v^-(e_k)$ could be in $S$.
  
Now, assume there is an approximation with ratio $\alpha := 1+\frac{2\beta}{(\Delta+1)^2}$ for \textsc{Maximum Independent Set} on 2-subdivisions running in subexponential time, where $1+\beta > 1$ is a ratio which is not attainable for \textsc{Maximum Independent Set} on graphs of maximum degree $\Delta$ according to Theorem~\ref{thm:inapprox-mis}.
On instance $H$, this algorithm would output a solution with $m'$ edge vertices and $s'$ original vertices.
As we already observed this solution can be easily (in polynomial time) transformed into an at-least-as-good solution with $m$ edge vertices and $s''$ original vertices forming an independent set in $G$. Further, we may assume that  $s'' \geqslant n / (\Delta+1)$ since for any independent set of $G$, we can obtain an independent set of $H$ consisting of the same set of original vertices and $m$ edge vertices. 
Since $m \leqslant n \Delta / 2$ and $s'' \geqslant n / (\Delta+1)$, we obtain $m \leqslant s'' \Delta(\Delta+1)/2$ and $2m/(\Delta+1)^2 \leqslant s''\Delta /(\Delta+1)$.
From $\frac{m+s}{m+s''} \leqslant \alpha$ and $\Delta \geqslant 3$, we have 
\[ s	\leqslant m\cdot \frac{2\beta}{(\Delta+1)^2} + s''\cdot (1+ \frac{2\beta}{(\Delta+1)^2})
	\leqslant s'' (\frac{\Delta\beta}{\Delta+1} + 1 +\frac{2\beta}{(\Delta+1)^2} )
	\leqslant s'' (1+\beta)
\]
  This contradicts the inapproximability of Theorem~\ref{thm:inapprox-mis}.
Indeed, note that the number of vertices of $H$ is only a constant times the number of vertices of $G$ (recall that $G$ has bounded maximum degree, hence $m=O(n)$).
\end{proof}

Recalling that independent set is a clique in the complement, we get the following.
\begin{corollary}\label{cor:hardness-co2subd}
  There is a constant $\alpha > 1$ such that for any $\varepsilon > 0$, \cli on the class of all the co-2-subdivisions has no $\alpha$-approximation algorithm running in subexponential time $2^{n^{1-\varepsilon}}$, unless the ETH fails.
\end{corollary}

For exact algorithms the subexponential time that we rule out under the ETH is not only $2^{n^{1-\varepsilon}}$ but actually any $2^{o(n)}$.

Now, to Theorem~\ref{thm:hardness-filled-ellipses} and Theorem~\ref{thm:hardness-filled-triangles}, it is sufficient to show that intersection graphs of (filled) ellipses or of (filled) triangles contain all co-2-subdivisions.
We start with (filled) triangles since the construction is straightforward.

\begin{lemma}\label{lem:triangles-co2subd}
The class of intersection graphs of filled triangles contains all co-2-subdivisions.
\end{lemma}
  
\begin{proof}
Let $G$ be any graph with $n$ vertices $v_1, \ldots, v_n$ and $m$ edges $e_1,\ldots,e_m$, and $H$ be its co-2-subdivision.  
  We start with $n+2$ points $p_0, p_1, p_2, \ldots, p_n, p_{n+1}$ forming a convex monotone chain.
  Those points can be chosen as $p_i := (i,p(i))$ where $p$ is the equation of a positive parabola taking its minimum at $(0,0)$.
  For each $i \in [0,n+1]$, let $q_i$ be the reflection of $p_i$ by the line of equation $y = 0$.
  Let $x := (n+1,0)$.
  For each vertex $v_i \in V(G)$ the filled triangle $\delta_i := p_iq_ix$ encodes $v_i$.
  Observe that the points $p_0=q_0$, $p_{n+1}$, and $q_{n+1}$ will only be used to define the filled triangles encoding edges. 

  To encode (the two new vertices of) a subdivided edge $e_k=v_iv_j$, we use two filled triangles $\Delta^+_k$ and $\Delta^-_k$. The triangle 
  $\Delta^+_k$ (resp. $\Delta^-_k$) has an edge which is supported by $\ell(p_{i-1},p_{i+1})$ (resp. $\ell(q_{j-1},q_{j+1})$) and is prolonged so that it crosses the boundary of each $\delta_{i'}$ but $\delta_i$ (resp. but $\delta_j$).
  A second edge of $\Delta^+_k$ and $\Delta^-_k$ are parallel and make with the horizontal a small angle $\varepsilon k$, where $\varepsilon> 0$ is chosen so that $\varepsilon m$ is smaller than the angle formed by $\ell(p_0,p_1)$ with the horizontal line.
  Those almost horizontal edges intersect for each pair $\Delta^+_{k'}$ and $\Delta^-_{k''}$ with $k' \neq k''$ intersects close to the same point.
  Filled triangles $\Delta^+_k$ and $\Delta^-_k$ do not intersect.
  See Figure~\ref{fig:triangles-co2subd} for the complete picture.

    It is easy to check that the intersection graph of $\{\delta_i\}_{i \in [n]} \cup \{\Delta^+_k,\Delta^-_k\}_{k \in [m]}$ is $H$.
    The family $\{\delta_i\}_{i \in [n]}$ forms a clique since they all contain for instance the point $x$.
    The filled triangle $\Delta^+_k$ (resp. $\Delta^-_k$) intersects every other filled triangles except $\Delta^-_k$ (resp. $\Delta^+_k$) and $\delta_i$ (resp. $\delta_j$) with $e_k=v_iv_j$.
    
    One may observe that no triangle is fully included in another triangle.
    So the construction works both as the intersection graph of filled triangles \emph{and} triangles without their interior.
    The edge of a $\Delta^+_k$ or a $\Delta^-_k$ crossing the boundary of all but one $\delta_i$, and the almost horizontal edge can be arbitrary prolonged to the right and to the left respectively.
    Thus, the triangles can all be made isosceles. 
   \end{proof}
  
    \begin{figure}[h!]
      \centering
      \begin{tikzpicture}[
          inv/.style={opacity=0},
          dot/.style={fill,circle,inner sep=-0.01cm},
          vert/.style={draw, fill=red, opacity=0.2},
          verta/.style={draw, fill=blue, opacity=0.2},
          vertb/.style={draw, fill=green, opacity=0.2},
          extended line/.style={shorten >=-#1,shorten <=-#1},
          extended line/.default=1cm,
          one end extended/.style={shorten >=-#1},
          one end extended/.default=1cm,
        ]
        
        \def\n{5}
        \def\a{0.08}
        \def\b{0.2}
        \def\c{0.02}
        \def\d{-5}
        \def\e{5}
        \def\f{0.05}
        \coordinate (su) at (\d,\c) ;
        \coordinate (sd) at (\d,-\c) ;
        \coordinate (eu) at (\e,\c) ;
        \coordinate (ed) at (\e,-\c) ;

        \coordinate (su2) at (\d,-0.1) ;
        \coordinate (eu2) at (\e,0.4) ;
        \coordinate (sd2) at (\d,-0.1 - 2 * \c) ;
        \coordinate (ed2) at (\e,0.4 - 2 * \c) ;
        \coordinate (x) at (\n+1,0) ;
        \node[dot] at (x) {} ;
        
        \foreach \i in {1,...,5}{
          \coordinate (p\i) at (\i,\i * \i * \a + \i * \b) ;
          \coordinate (q\i) at (\i,- \i * \i * \a - \i * \b) ;
          \coordinate (pd\i) at (\i,\i * \i * \a + \i * \b - \f) ;
          \coordinate (qu\i) at (\i,- \i * \i * \a - \i * \b + \f) ;
        }
        \foreach \i in {1,...,5}{
          \node[dot] at (p\i) {} ;
          \node[dot] at (q\i) {} ;
          \draw[vert] (p\i) -- (q\i) -- (x) -- cycle ;
        }

        \path[name path=J1,overlay] (su) -- (eu)--([turn]0:5cm);
        \path[name path=K1,overlay] (pd2) -- (0,0)--([turn]0:5cm);
        \path[name path=J2,overlay] (sd) -- (ed)--([turn]0:5cm);
        \path[name path=K2,overlay] (qu5) -- (qu3)--([turn]0:5cm);

        \path [name intersections={of=J1 and K1,by={I1}}];
        \path [name intersections={of=J2 and K2,by={I2}}];

        \coordinate (e1) at (I1) ;
        \coordinate (e2) at (I2) ;

        \coordinate (c1) at ( $ (pd2)!-2!(e1) $ ) ;
        \coordinate (c2) at ( $ (qu5)!-0.1!(e2) $ ) ;

        \path[name path=J3,overlay] (su2) -- (eu2)--([turn]0:5cm);
        \path[name path=K3,overlay] (pd3) -- (pd1)--([turn]0:5cm);
        \path[name path=J4,overlay] (sd2) -- (ed2)--([turn]0:5cm);
        \path[name path=K4,overlay] (qu4) -- (qu2)--([turn]0:5cm);

        \path [name intersections={of=J3 and K3,by={I3}}];
        \path [name intersections={of=J4 and K4,by={I4}}];

        \coordinate (e3) at (I3) ;
        \coordinate (e4) at (I4) ;
        \coordinate (c3) at ( $ (pd3)!-1.4!(e3) $ ) ;
        \coordinate (c4) at ( $ (qu4)!-0.35!(e4) $ ) ;

         \foreach \i/\j/\k in {1/su/vertb,2/sd/vertb,3/su2/verta,4/sd2/verta}{
            \draw[\k] (\j) -- (c\i) -- (e\i) -- cycle ;
        }
        
      \end{tikzpicture}
      \caption{A co-2-subdivision of a graph with $5$ vertices (in red) represented with triangles. Only two edges are shown: one between vertices $1$ and $4$ (green) and one between vertices $2$ and $3$ (blue).}
      \label{fig:triangles-co2subd}
    \end{figure}
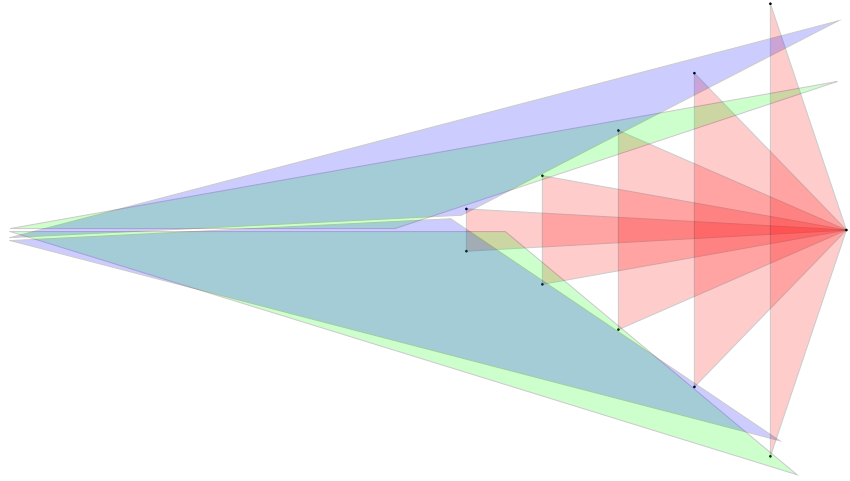

  We use the same ideas for the construction with filled ellipses.
  The two important sides of a triangle encoding an edge of the initial graph $G$ become two tangents of the ellipse.  

  \begin{lemma}\label{lem:ellipses-co2subd}
    The class of intersection graphs of filled ellipses contains all co-2-subdivisions.
  \end{lemma}

  \begin{proof}
    Let $G$ be any graph with $n$ vertices $v_1, \ldots, v_n$ and $m$ edges $e_1,\ldots,e_m$, and $H$ be its co-2-subdivision. 
    We start with the convex monotone chain $p_0, p_1, p_2, \ldots, p_{n-1}, p_n, p_{n+1}$, only the gap between $p_i$ and $p_{i+1}$ is chosen very small compared to the positive $y$-coordinate of $p_0$.
    The disks $\mathcal D_i$ encoding the vertices $v_i \in G$ must form a clique.
    We also take $p_0$ with a large $x$-coordinate.
    For $i \in [0,n+1]$, $q_i$ is the symmetric of $p_i$ with respect to the $x$-axis.
    For each $i \in [n]$, we define $\mathcal D_i$ as the disk whose boundary is the unique circle which goes through $p_i$ and $q_i$, and whose tangent at $p_i$ has the direction of $\ell(p_{i-1},p_{i+1})$.
    It can be observed that, by symmetry, the tangent of $\mathcal D_i$ at $q_i$ has the direction of $\ell(q_{i-1},q_{i+1})$.

    Let us call $\tau^+_i$ (resp. $\tau^-_i$) the tangent of $\mathcal D_i$ at $p_i$ (resp. at $q_i$) very slightly translated upward (resp. downward).
    The tangent $\tau^+_i$ (resp. $\tau^-_i$) intersects every disks $\mathcal D_{i'}$ but $\mathcal D_i$ (see Figure~\ref{fig:disks-all-but-one}).
    Let denote by $p'_i$ (resp. $q'_i$) be the projection of $p_i$ (resp. $q_i$) onto $\tau^+_i$ (resp. onto $\tau^-_i$) 
\begin{figure}
  \centering
  \begin{tikzpicture}[scale=1.7, xscale=-1,
      vert/.style={draw, fill=red, opacity=0.2},
      dot/.style={fill,circle,inner sep=-0.03cm},
    ]
    \def\s{0.5}
    \def\t{0.045}
    \foreach \i in {0,...,3}{
      \coordinate (c\i) at (- \i * \s,0) ;
      \pgfmathsetmacro\r{1+\i * \i * \t} ;
      \draw[vert] (c\i) circle (\r) ;
      \pgfmathsetmacro\a{90.1 - \i * 10}
      \pgfmathsetmacro\x{- \i * \s + \r * cos(\a)}
      \pgfmathsetmacro\y{ \r * sin(\a)}
      \pgfmathsetmacro\ym{- \r * sin(\a)}
      \pgfmathsetmacro\j{1.1 + \i * 0.4}
      \pgfmathsetmacro\jj{2.6 - \i * 0.2}
      \pgfmathtruncatemacro\ipo{\i+1}
      \coordinate (p\i) at (\x, \y) ;
      \node at (\x, \y+0.1) {$p_\ipo$} ;
      \coordinate (e\i) at (\x, \ym) ;
      \node at (\x, \ym-0.1) {$q_\ipo$};
      \path[overlay] (c\i) -- (p\i) -- ([turn]-90:\j cm) node (q\i) {} ;
      \path[overlay] (c\i) -- (p\i) -- ([turn]90:\jj cm) node (qq\i) {} ;
      \path[overlay] (c\i) -- (e\i) -- ([turn]90:\j cm) node (f\i) {} ;
      \path[overlay] (c\i) -- (e\i) -- ([turn]-90:\jj cm) node (ff\i) {} ;
    }
    \node at (c2) {$\mathcal D_3$} ; 
    \draw[blue,very thick] ($ (q2) + (0,-0.02) $) -- ($ (qq2) + (0,0.05) $) ;
    \draw[red,very thick] (c2) circle (1+4*\t) ;

    \node at (-2,1.85) {$\tau^+_3$}; 

    \foreach \i in {0,...,3}{
      \node[dot] at (p\i) {} ;
      \node[dot] at (e\i) {} ;
    }
  \end{tikzpicture}
   \caption{The blue line intersects every red disk but the third one.}
  \label{fig:disks-all-but-one}
\end{figure}
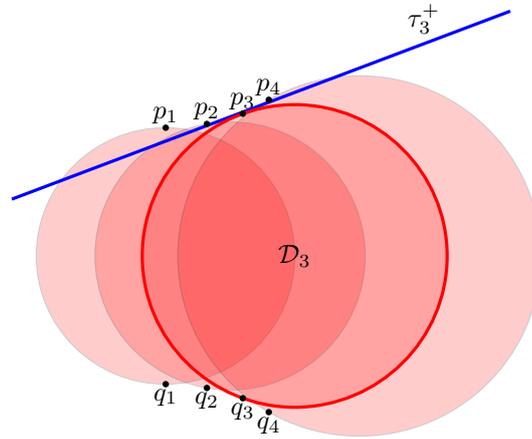
For each $k \in [m]$, let $\ell_k$ be the line crossing the origin $O=(0,0)$ and forming with the horizontal an angle $\varepsilon k$, where $\varepsilon k$ is smaller than the angle formed by $\ell(p_0,p_1)$ with the horizontal.
Let $\ell^+_k$ (resp. $\ell^-_k$) be $\ell_k$ very slightly translated upward (resp. downward). 
  To encode an edge $e_k=v_iv_j$, we have two filled ellipses $\mathcal E^+_k$ and $\mathcal E^-_k$. The ellipse
  $\mathcal E^+_k$ (resp. $\mathcal E^-_k$) is defined as being tangent with $\tau^+_i$ at $p'_i$ (resp. with $\tau^-_j$ at $q'_j$)
  and tangent at $\ell^+_k$ (resp. $\ell^-_k$) at the point of $x$-coordinate $0$ (thus very close to $O$), where $e_k=v_iv_j$.
  The proof that the intersection graph of $\{\mathcal D_i\}_{i \in [n]} \cup \{\mathcal E^+_k,\mathcal E^-_k\}_{k \in [m]}$ is $H$ is similar to the case of filled triangles.

  As no ellipse is fully contained in another ellipse, this construction works for both filled ellipses \emph{and} ellipses without their interior.
  
  We place $p_0$ at $P:=(\sqrt 3/2,1/2)$ and make the distance between $p_i$ and $p_{i+1}$ very small compared to 1.
  All points $p_i$ are very close to $P$ and all points $q_i$ are very close to $Q:=(\sqrt 3/2,-1/2)$.
  This makes the radius of all disks $\mathcal D_i$ arbitrarily close to 1.
  We choose the convex monotone chain $p_0, \ldots, p_{n+1}$ so that $\ell(p_0,p_1)$ forms a 60-degree angle with the horizontal.
As, the chain is strictly convex but very close to a straight-line, $\ell(p_0,p_1) \approx \ell(p_n,p_{n+1}) \approx \ell(p_i,p_{i+1}) \approx \ell(p_i,p_{i+2})$. Thus, all those lines almost cross $P$ and form an angle of roughly 60-degree with the horizontal.
  The same holds for points $q_i$.
  For the choice of an elliptical disk tangent to the $x$-axis at $O$ and to a line with a 60-degree slope at $P$ (resp. at $Q$), we take a disk of radius 1 centered at $(0,1)$ (resp. at $(0,-1)$); see Figure~\ref{fig:almost-disks}.
  
  \begin{figure}[h!]
  \centering
  \begin{tikzpicture}[
      scale=1.4,
      vert/.style={draw, fill=red, opacity=0.2},
      verta/.style={draw, fill=blue, opacity=0.2},
      vertb/.style={draw, fill=blue, opacity=0.2},
      dot/.style={fill,circle,inner sep=-0.03cm}]
      \draw[vertb] (0,0) circle (1) ;
      \node at (0,0) {$\mathcal E^-_k$} ;
      \coordinate (cm) at (0,0) ;
      \draw[verta] (0,2) circle (1) ;
      \node at (0,2) {$\mathcal E^+_k$} ;
      \coordinate (cp) at (0,2) ;
      \draw[vert] (1.73,1) circle (1) ;
      \node at (1.73,1) {$\mathcal D_i$} ;
      \coordinate (cv) at (1.73,1) ;
      \node[dot] at (1.73/2,1.5) {} ;
      \node[dot] at (1.73/2,0.5) {} ;
      \node[dot] at (0,1) {} ;
      \node at (1.73/2,1.75) {$P$} ;
      \node at (1.73/2,0.25) {$Q$} ;
      \node at (-0.25,1) {$O$} ;
      \coordinate (P) at (1.73/2,1.5) ;
      \coordinate (Q) at (1.73/2,0.5) ;
      \coordinate (O) at (0,1) ;
      \draw[opacity=0.3] (O) -- (P) -- (Q) -- cycle ;
      \draw[opacity=0.3] (cm) -- (O) -- (cp) -- (P) -- (cv) -- (Q) -- cycle ;
    \end{tikzpicture}
    \caption{The layout of the disks $\mathcal D_i$, and the elliptical disks $\mathcal E^+_k$ and $\mathcal E^-_k$.}
    \label{fig:almost-disks}
  \end{figure}
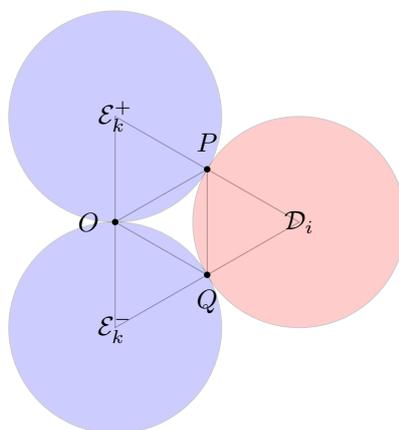

  The acute angle formed by $\ell_1$ and $\ell_m$ (incident in $O$) is made arbitrarily small so that, by continuity of the elliptical disk defined by two tangents at two points, the filled ellipses $\mathcal E^+_k$ and $\mathcal E^-_k$ have eccentricity arbitrarily close to 0 and major axis arbitrarily close to 1. 
  \end{proof}

  In the construction, we made \emph{both} the eccentricity of the (filled) ellipses arbitrarily close to 0 and the ratio between the largest and the smallest major axis arbitrarily close to 1.
  We know that this construction is very unlikely to work for the extreme case of unit disks, since a polynomial algorithm is known for \textsc{Max Clique}.
  Note that even with disks of arbitrary radii, Theorem~\ref{thm:main-structural-non-disk} unconditionally proves that the construction does fail.
  Indeed the co-2-subdivision of $C_3+C_3$ is the complement of $C_9+C_9$, hence not a disk graph.

\subsection{Homothets of a convex polygon}

Another natural direction of generalizing a result on disk intersection graphs is to consider {\em pseudodisk intersection graphs}, i.e., intersection graphs of collections of closed subsets of the plane (regions bounded by simple Jordan curves) that are pairwise in a {\em pseudodisk} relationship (see Kratochv\'il \cite{DBLP:conf/gd/Kratochvil96}). Two regions $A$ and $B$ are in pseudodisk relation if both differences $A\setminus B$ and $B\setminus A$ are arc-connected.
It is known that $P_{hom}$ graphs, i.e., intersection graphs of homothetic copies of a fixed polygon $P$, are pseudodisk intersection graphs~\cite{agarwal453state}. As shown by Brimkov {\em et al.}, for every convex $k$-gon $P$, a $P_{hom}$ graph with $n$ vertices has at most $n^k$ maximal cliques~\cite{DBLP:journals/corr/BrimkovJKKPRST14}. This clearly implies that \cli, but also \textsc{Clique $p$-Partition} for fixed $p$ is polynomially solvable in $P_{hom}$ graphs.
Actually, the bound on the maximum number of maximal cliques from \cite{DBLP:journals/corr/BrimkovJKKPRST14} holds for a more general class of graphs, called $k_{DIR}$-CONV, which  admit a intersection representation by convex polygons, whose every side is parallel to one of $k$ directions.

Moreover, we observe that Theorem \ref{thm:coEvenCycles} cannot be generalized to $P_{hom}$ graphs or $k_{DIR}$-CONV graphs. Indeed, consider the complement $\overline{P_n}$ of an $n$-vertex path $P_n$. The number of maximal cliques in $\overline{P_n}$, or, equivalently, maximal independent sets  in $P_n$ is $\Theta(c^n)$ for $c \approx 1.32$, i.e., exponential in $n$ \cite{DBLP:journals/jgt/Furedi87}. Therefore, for every fixed polygon $P$ (or for every fixed $k$) there is $n$, such that $\overline{P_n}$ is not a $P_{hom}$ ($k_{DIR}$-CONV) graph.

\section{Perspectives}\label{sec:perspectives}

We presented the first QPTAS and subexponential algorithm for \cli on disk graphs.
Our subexponential algorithm extends to the weighted case and yields a polynomial algorithm if both the degree $\Delta$ and the odd girth $c$ of the complement graph are constant.
Indeed, our full characterization of disk graphs with co-degree 2, implies a backdoor-to-bipartiteness of size $c\Delta$ in the complement.

We have also paved the way for a potential NP-hardness construction.
We showed why the versatile approach of representing complements of even subdivisions of graphs forming a class on which \mis is NP-hard fails if the class is \emph{general graphs}, \emph{planar graphs}, or even any class containing the disjoint union of two odd cycles.
This approach was used by Middendorf for some string graphs \cite{Middendorf92} (with the class of all graphs), Cabello et al. \cite{CabelloCL13} to settle the then long-standing open question of the complexity of \cli for segments (with the class of planar graphs), in Section~\ref{sec:gen&lim} of this paper for ellipses and triangles (with the class of all graphs).
Determining the complexity of \mis on graphs without two vertex-disjoint odd cycles as an induced subgraph is a valuable first step towards settling the complexity of \cli on disks.

Another direction is to try and strengthen our QPTAS in one of two ways: either to obtain a PTAS for \cli on disk graphs, or to obtain a QPTAS (or PTAS) for \wcli on disk graphs.
It is interesting to note that Bock et al. \cite{Bock14} showed a PTAS for \mwis for graphs $G$ with $\ocp(G) = O(\log n / \log \log n)$. However, this bound is too weak to use a win-win approach similar to Theorem \ref{thm:qptas}.

\bibliographystyle{abbrv}

\begin{thebibliography}{10}

\bibitem{agarwal453state}
P.~K. Agarwal, J.~Pach, and M.~Sharir.
\newblock {State of the Union (of Geometric Objects)}.
\newblock {\em Surveys in Discrete and Computational Geometry: Twenty Years
  Later. Contemporary Mathematics}, 453:9--48, 2008.

\bibitem{AlberF04}
J.~Alber and J.~Fiala.
\newblock Geometric separation and exact solutions for the parameterized
  independent set problem on disk graphs.
\newblock {\em J. Algorithms}, 52(2):134--151, 2004.

\bibitem{AlonYZ97}
N.~Alon, R.~Yuster, and U.~Zwick.
\newblock Finding and counting given length cycles.
\newblock {\em Algorithmica}, 17(3):209--223, 1997.

\bibitem{Ambuhl05}
C.~Amb{\"{u}}hl and U.~Wagner.
\newblock {The Clique Problem in Intersection Graphs of Ellipses and
  Triangles}.
\newblock {\em Theory Comput. Syst.}, 38(3):279--292, 2005.

\bibitem{Atminas16}
A.~Atminas and V.~Zamaraev.
\newblock On forbidden induced subgraphs for unit disk graphs.
\newblock {\em arXiv preprint arXiv:1602.08148}, 2016.

\bibitem{bang2006}
J.~Bang-Jensen, B.~Reed, M.~Schacht, R.~{\v{S}}{\'a}mal, B.~Toft, and
  U.~Wagner.
\newblock On six problems posed by {Jarik Ne{\v{s}}et{\v{r}}il}.
\newblock {\em Topics in Discrete Mathematics}, pages 613--627, 2006.

\bibitem{Bock14}
A.~Bock, Y.~Faenza, C.~Moldenhauer, and A.~J. Ruiz{-}Vargas.
\newblock {Solving the Stable Set Problem in Terms of the Odd Cycle Packing
  Number}.
\newblock In {\em 34th International Conference on Foundation of Software
  Technology and Theoretical Computer Science, {FSTTCS} 2014, December 15-17,
  2014, New Delhi, India}, pages 187--198, 2014.

\bibitem{Bonnet15}
{\'{E}}.~Bonnet, B.~Escoffier, E.~J. Kim, and V.~Th.~Paschos.
\newblock {On Subexponential and FPT-Time Inapproximability}.
\newblock {\em Algorithmica}, 71(3):541--565, 2015.

\bibitem{Brandstadt1999}
A.~Brandst{\"a}dt, V.~B. Le, and J.~P. Spinrad.
\newblock {\em Graph classes: a survey}.
\newblock SIAM, 1999.

\bibitem{Breu98}
H.~Breu and D.~G. Kirkpatrick.
\newblock {Unit disk graph recognition is NP-hard}.
\newblock {\em Comput. Geom.}, 9(1-2):3--24, 1998.

\bibitem{DBLP:journals/corr/BrimkovJKKPRST14}
V.~E. Brimkov, K.~Junosza{-}Szaniawski, S.~Kafer, J.~Kratochv{\'{\i}}l,
  M.~Pergel, P.~Rz{\k{a}}{\.{z}}ewski, M.~Szczepankiewicz, and J.~Terhaar.
\newblock Homothetic polygons and beyond: Intersection graphs, recognition, and
  maximum clique.
\newblock {\em CoRR}, abs/1411.2928, 2014.

\bibitem{CabelloOpen}
S.~Cabello.
\newblock Maximum clique for disks of two sizes.
\newblock Open problems from {Geometric Intersection Graphs: Problems and
  Directions CG Week Workshop, Eindhoven, June 25, 2015}
  (\url{http://cgweek15.tcs.uj.edu.pl/problems.pdf}), 2015.
\newblock [Online; accessed 07-December-2017].

\bibitem{Cabello2015}
S.~Cabello.
\newblock {Open problems presented at the Algorithmic Graph Theory on the
  Adriatic Coast workshop, Koper, Slovenia}
  (\url{https://conferences.matheo.si/event/6/picture/35.pdf}).
\newblock June 16-19 2015.

\bibitem{CabelloCL13}
S.~Cabello, J.~Cardinal, and S.~Langerman.
\newblock The clique problem in ray intersection graphs.
\newblock {\em Discrete {\&} Computational Geometry}, 50(3):771--783, 2013.

\bibitem{ceroi}
S.~Ceroi.
\newblock {The clique number of unit quasi-disk graphs}.
\newblock Technical Report RR-4419, {INRIA}, Mar. 2002.

\bibitem{Chan2003}
T.~M. Chan.
\newblock Polynomial-time approximation schemes for packing and piercing fat
  objects.
\newblock {\em J. Algorithms}, 46(2):178--189, 2003.

\bibitem{Chlebik06}
M.~Chleb{\'{\i}}k and J.~Chleb{\'{\i}}kov{\'{a}}.
\newblock Complexity of approximating bounded variants of optimization
  problems.
\newblock {\em Theor. Comput. Sci.}, 354(3):320--338, 2006.

\bibitem{Clark90}
B.~N. Clark, C.~J. Colbourn, and D.~S. Johnson.
\newblock Unit disk graphs.
\newblock {\em Discrete Mathematics}, 86(1-3):165--177, 1990.

\bibitem{Erlebach2005}
T.~Erlebach, K.~Jansen, and E.~Seidel.
\newblock Polynomial-time approximation schemes for geometric intersection
  graphs.
\newblock {\em {SIAM} J. Comput.}, 34(6):1302--1323, 2005.

\bibitem{DBLP:conf/waoa/Fishkin03}
A.~V. Fishkin.
\newblock Disk graphs: {A} short survey.
\newblock In K.~Jansen and R.~Solis{-}Oba, editors, {\em Approximation and
  Online Algorithms, First International Workshop, {WAOA} 2003, Budapest,
  Hungary, September 16-18, 2003, Revised Papers}, volume 2909 of {\em Lecture
  Notes in Computer Science}, pages 260--264. Springer, 2003.

\bibitem{DBLP:journals/jgt/Furedi87}
Z.~F{\"{u}}redi.
\newblock The number of maximal independent sets in connected graphs.
\newblock {\em Journal of Graph Theory}, 11(4):463--470, 1987.

\bibitem{Gyori97}
E.~Györi, A.~V. Kostochka, and T.~Łuczak.
\newblock Graphs without short odd cycles are nearly bipartite.
\newblock {\em Discrete Mathematics}, 163(1):279 -- 284, 1997.

\bibitem{ImpagliazzoETH}
R.~Impagliazzo, R.~Paturi, and F.~Zane.
\newblock Which problems have strongly exponential complexity?
\newblock {\em Journal of Computer and System Sciences}, 63(4):512--530, Dec.
  2001.

\bibitem{Kang12}
R.~J. Kang and T.~M{\"{u}}ller.
\newblock {Sphere and Dot Product Representations of Graphs}.
\newblock {\em Discrete {\&} Computational Geometry}, 47(3):548--568, 2012.

\bibitem{DBLP:conf/stoc/KawarabayashiR10}
K.~Kawarabayashi and B.~A. Reed.
\newblock Odd cycle packing.
\newblock In {\em Proceedings of the 42nd {ACM} Symposium on Theory of
  Computing, {STOC} 2010, Cambridge, Massachusetts, USA, 5-8 June 2010}, pages
  695--704, 2010.

\bibitem{Keller17}
C.~Keller, S.~Smorodinsky, and G.~Tardos.
\newblock {On Max-Clique for intersection graphs of sets and the
  Hadwiger-Debrunner numbers}.
\newblock In {\em Proceedings of the Twenty-Eighth Annual {ACM-SIAM} Symposium
  on Discrete Algorithms, {SODA} 2017, Barcelona, Spain, Hotel Porta Fira,
  January 16-19}, pages 2254--2263, 2017.

\bibitem{precoloring}
J.~Kratochvil.
\newblock Precoloring extension with fixed color bound.
\newblock 62:139--153, 1993.

\bibitem{DBLP:conf/gd/Kratochvil96}
J.~Kratochv{\'{\i}}l.
\newblock Intersection graphs of noncrossing arc-connected sets in the plane.
\newblock In {\em Graph Drawing, Symposium on Graph Drawing, {GD} '96,
  Berkeley, California, USA, September 18-20, Proceedings}, pages 257--270,
  1996.
  
  
  
\bibitem{DBLP:journals/jct/KratochvilM94}
J.~Kratochv{\'{\i}}l and J.~Matou{\v{s}}ek.
\newblock Intersection graphs of segments.
\newblock {\em J. Comb. Theory, Ser. {B}}, 62(2):289--315, 1994.


\bibitem{Marx08}
D.~Marx.
\newblock {Parameterized Complexity and Approximation Algorithms}.
\newblock {\em Comput. J.}, 51(1):60--78, 2008.

\bibitem{McKee1999}
T.~A. McKee and F.~R. McMorris.
\newblock {\em Topics in intersection graph theory}.
\newblock SIAM, 1999.

\bibitem{Middendorf92}
M.~Middendorf and F.~Pfeiffer.
\newblock The max clique problem in classes of string-graphs.
\newblock {\em Discrete Mathematics}, 108(1-3):365--372, 1992.

\bibitem{Moshkovitz10}
D.~Moshkovitz and R.~Raz.
\newblock Two-query {PCP} with subconstant error.
\newblock {\em J. {ACM}}, 57(5):29:1--29:29, 2010.

\bibitem{Raghavan03}
V.~Raghavan and J.~P. Spinrad.
\newblock {Robust algorithms for restricted domains}.
\newblock {\em J. Algorithms}, 48(1):160--172, 2003.

\bibitem{Tait1877}
P.~G. Tait.
\newblock Some elementary properties of closed plane curves.
\newblock {\em Messenger of Mathematics, New Series}, (69):270--272, 1877.

\bibitem{EJvL2009}
E.~J. van Leeuwen.
\newblock {\em Optimization and Approximation on Systems of Geometric Objects}.
\newblock PhD thesis, Utrecht University, 2009.

\end{thebibliography}

\end{document}